\documentclass[copyright,creativecommons]{eptcs}
\providecommand{\event}{MSFP 2014} 
\newif\iffull\fullfalse
\newif\ifdraft\draftfalse
\usepackage{url}
\usepackage{amsmath}
\usepackage{amssymb}
\usepackage{stmaryrd}
\usepackage{semantic}
\usepackage{latexsym}
\usepackage{amsthm}
\usepackage{listings}
\usepackage{color}
\theoremstyle{plain}

\newtheorem{theorem}[]{Theorem}
\newtheorem{lemma}[theorem]{Lemma}

\newtheorem{definition}[theorem]{Definition}

\newcommand{\code}[1]{\ensuremath{\kw{#1}}}
\newcommand{\sfont}[1]{\ensuremath{\mathsf{#1}}}

\newcommand{\tygt}{\geqslant}

\newcommand{\mbind}{\ensuremath{\mathsf{bind}}}
\newcommand{\munit}{\ensuremath{\mathsf{unit}}}

\newcommand{\rew}{\leadsto}

\newcommand{\slam}[2]{\lambda #1. #2}
\newcommand{\sif}[3]{\sfont{if}~#1~\sfont{then}~#2~\sfont{else}~#3}
\newcommand{\sapp}[2]{#1\;#2}

\newcommand{\slet}[3]{\sfont{let}\;#1\!=\!#2\;\sfont{in}\;{#3}}
\newcommand{\sletrec}[3]{\sfont{let rec}\;#1\!=\!#2\;\sfont{in}\;{#3}}

\newcommand{\tgtfont}[1]{\sfont{#1}}

\newcommand{\tgte}{\tgtfont{e}}

\newcommand{\eapp}[2]{#1\;#2}

\newcommand{\kstar}{\star}
\newcommand{\kmon}{\bullet}

\newcommand{\morph}[2]{#1 \hookrightarrow #2}
\newcommand{\mId}{\tfont{Bot}}

\newcommand{\mIST}{\ensuremath{\mbox{\textit{IST}}}}

\newcommand{\Hi}{\tfont{H}}
\newcommand{\Lo}{\tfont{L}}

\newcommand{\gm}[1][M]{\text{\sfont{#1}}}
\newcommand{\ty}{\tau}
\newcommand{\tscheme}[3]{\forall #1.\, #2 => #3}
\newcommand{\tfont}[1]{\mathit{#1}}

\newcommand{\tfun}[2]{#1 \rightarrow #2}

\newcommand{\tapp}[2]{#1\;#2}

\newcommand{\Gen}[2]{\ensuremath{\mbox{\textit{Gen}}(#1, #2)}}
\newcommand{\ftv}[1]{\ensuremath{\mbox{\textsf{ftv}}(#1)}}

\newcommand{\tint}{\tfont{int}}

\newcommand{\prefix}[2]{#1\,|\,#2 |-}

\newcommand{\atforbindingcolon}{\mathcode`@="003A}

\newcommand{\aset}[1]{\{#1\}}

\newcommand{\dom}[1]{\mbox{\textit{dom}}(#1)}
\newcommand{\nquad}{\ensuremath{\!\!\!\!}}
\newcommand{\nqquad}{\nquad\nquad}

\newcommand{\ratio}{.35}

\newcommand\mycase[1]{\noindent\textbf{Case #1}:}
\newcommand\myscase[1]{$\noindent$\textbf{Sub-case #1}:}

\newcommand\eqannot[1]{\mbox{(\textit{\color{gray}{#1}})}}
\newcommand\mconst[1][M]{\mathsf{\footnotesize{#1}}}
\newcommand\mvar{\rho}
\newcommand\tbind[2]{{#1} \rhd {#2}} 
\newcommand\bind[2]{\tbind{\mathsf{#1}}{\mathsf{#2}}}  

\newcommand{\hbind}{>\!>\!>\!\!\mbox{=}\strut}

\newcommand{\maybecolor}[1]{}

\def\lstCaml{\lstset{
  basicstyle=\footnotesize,
  morekeywords=[1]{type,val,fun,let,in,ref,of,try,if,then,else,match,with,do,open,module,rec,end,sig,private,when,while},
  morekeywords=[2]{public,interface,internal,extern,prop},
  morekeywords=[3]{},
  morekeywords=[4]{}, 
  morestring=[b]",
  sensitive=true,%
  numbersep=4pt,
  columns=[l]fullflexible,
  texcl=true,
  mathescape=true,
  identifierstyle={\rmfamily\sffamily\maybecolor{dkgreen}},
  keywordstyle=[1]{\small\bfseries\maybecolor{dkblue}},
  keywordstyle=[2]{\bfseries\maybecolor{dkblue}},
  keywordstyle=[3]{\small\bfseries},
  keywordstyle=[4]{\rmfamily\itshape},
  rangeprefix=(*---\ ,
  includerangemarker=false,
  stringstyle=\ttfamily,
  showspaces=false,
  moredelim=[is][\scriptsize\rmfamily\sffamily\maybecolor{dkblue}]{~}{~},
  comment=[l]{//},
  morecomment=[l]{\#},
  morecomment=[n]{(*}{*)},
  commentstyle={\itshape\color{gray}},
  literate={->}{$\rightarrow\,$}{1} 
           {''a}{{\small$\alpha$}}{1} 
           {''b}{{\small$\beta$}}{1} 
           {'a}{{\small$\alpha\, $}}{1} 
	   {'b}{{\small$\beta\, $}}{1} 
	   {'c}{{\small$\gamma\, $}}{1}
           {kstar}{{$\kstar\, $}}{1}
           {kmon}{{$\kmon\, $}}{1}
           {kprop}{{$\kprop\, $}}{1}
           {kafn}{{$\kafn\, $}}{1}
	   {lam}{$\lambda$}{1}
	   {Lam}{$\Lambda$}{1}
	   {exists}{$\exists$}{1}
	   {forall}{$\forall$}{1}
	   {Phi}{$\Phi\, $}{1}
	   {Psi}{$\Psi\, $}{1}
           {alpha}{$\alpha$}{1}
           {beta}{$\beta$}{1}
           {gamma}{$\gamma$}{1}
	   {phi}{{$\raisebox{.75pt}{$\phi$}$}}{1}
	   {psi}{{$\raisebox{.75pt}{$\psi$}$}}{1}
	   {bot}{$\bot$}{1}
	   {extlang}{$\lambda^{\mbox{\sc\tiny BX}}$}{1}
	   {iff}{$\Leftrightarrow\, $}{1}
     {>>>=}{$\hbind$}{2}
     {[[}{$\llbracket$}{2}
     {]]}{$\rrbracket$}{2}
	   {=>}{$\Rightarrow \ $}{2},
  breaklines=false}}

\lstCaml
\let\ls\lstinline
\newcommand{\kw}[1]{\mbox{\normalfont\lstinline{#1}}}

\usepackage{graphicx}
\usepackage[curve]{xypic}
\usepackage{wrapfig}

\definecolor{gray}{rgb}{0.3,0.3,0.3}
\definecolor{purple}{rgb}{0.63,0.13,0.94}
\definecolor{orange}{rgb}{1,0.5,0}
\definecolor{dkgreen}{rgb}{0,0.5,0}

\newcommand\gmb[1]{\ifdraft{\color{purple}{Gavin says: #1}}\fi}

\newcommand\mwh[1]{\ifdraft{\color{blue}{Mike says: #1}}\fi}

\newcommand{\eq}{\!\!=\!\!}
\renewcommand{\mId}{\ensuremath{\sfont{Id}}}
\newcommand\langext{\ensuremath{\lambda{\mbox{\textsc{\underline{pm}}}}}}
\newcommand\lang{\langext}

\atforbindingcolon

\iffull
\title{Polymonadic Programming\\(Extended version)}
\else
\title{Polymonadic Programming}
\fi

\author{Michael Hicks$^1$ ~~~ Gavin Bierman$^3$ ~~~ Nataliya Guts$^1$ ~~~ Daan Leijen$^2$ ~~~ Nikhil Swamy$^2$
      \institute{$~^1$University of Maryland, College Park ~~~~ $~^2$Microsoft Research ~~~~ $~^3$Oracle Labs}
      \email{mwh@cs.umd.edu, gavin.bierman@oracle.com, \{daan,nswamy\}@microsoft.com, nyuguts@gmail.com}}

\iffull
\def\titlerunning{Polymonadic Programming (Extended version)}
\else
\def\titlerunning{Polymonadic Programming}
\fi
\def\authorrunning{M. Hicks, G. Bierman, N. Guts, D. Leijen \& N. Swamy}

\newcommand{\tsend}[2]{\tfont{send}\,#1\,#2}
\newcommand{\trecv}[2]{\tfont{recv}\,#1\,#2}
\newcommand{\tdone}{\tfont{end}}

\newcommand\Binds{\ensuremath{P}}
\newcommand\hide[1]{\ensuremath{\mbox{\textit{Hide}}(#1)}}
\newcommand\simp[1][\bar\mvar]{\xrightarrow{\mbox{\tiny{simplify}}(#1)}}

\newcommand\tup{{\rotatebox[origin=c]{90}{$\rhd$}}}

\newcommand{\cat}[1]{\mathbb{#1}}

\newcommand\mconstrs{\mathcal{M}}

\newcommand{\xybind}[3]
{\xymatrix@R=0em@C=0em{
  \text{$#1$} \ar@{-}[dr] & & \text{$#2$}\ar@{-}[dl]\\
  & \ar[d]& \\
  & \text{$#3$} 
}}
\newcommand\fomega{F$\omega$}
\newcommand\flowsTo[2]{\textit{flowsTo}_{#2}{~#1}}
\newcommand\flowsFrom[2]{\textit{flowsFrom}_{#2}{~#1}}

\newcommand{\mymap}{\kw{map}}
\newcommand{\myunit}{\kw{unit}}
\newcommand{\mybind}{\kw{bind}}
\newcommand{\mylift}{\kw{lift}}

\begin{document}
\setlength{\textfloatsep}{0.75\textfloatsep}
\setlength{\floatsep}{\floatsep}
\linespread{0.9}

\pagestyle{plain}
\maketitle
\thispagestyle{plain}
\begin{abstract} 
  Monads are a popular tool for the working functional programmer to
  structure effectful computations. This paper presents
  \emph{polymonads}, a generalization of monads. Polymonads give the
  familiar monadic bind the more general type 
  \ls!forall $a,b$. L $a$ -> ($a$ -> M $b$) -> N $b$!, to
  compose computations with three different kinds of effects, rather
  than just one. Polymonads subsume monads and parameterized monads,
  and can express other constructions, including precise
  type-and-effect systems and information flow tracking; more
  generally, polymonads correspond to Tate's \emph{productoid}
  semantic model. We show how to equip a core language (called \lang)
  with syntactic support for programming with polymonads. Type
  inference and elaboration in \lang{} allows programmers to write
  polymonadic code directly in an ML-like syntax---our algorithms
  compute principal types and produce elaborated programs wherein
  the binds appear explicitly. Furthermore, we prove that the
  elaboration is \emph{coherent}: no matter which (type-correct) binds
  are chosen, the elaborated program's semantics will be the
  same. Pleasingly, the inferred types are easy to read: the polymonad
  laws justify (sometimes dramatic) simplifications, but with no
  effect on a type's generality.
\end{abstract}



\section{Introduction}
\label{sec:intro}





Since the time that Moggi first connected them to effectful
computation~\cite{moggi89computational}, \emph{monads} have proven to be a
surprisingly versatile computational structure.  Perhaps best known as
the foundation of Haskell's support for state, I/O, and other effects,
monads have also been used to structure APIs for libraries that
implement a wide range of programming tasks, including
parsers~\cite{hutton98monadic}, probabilistic
computations~\cite{ramsey02stochastic}, and functional
reactivity~\cite{elliot97functional,cooper06father}.


Monads (and morphisms between them) are not a panacea, however, and
so researchers have proposed various extensions.  Examples include
Wadler and Thiemann's~\cite{wadler:2003} indexed monad for typing
effectful computations; Filli{\^a}tre's
generalized monads~\cite{filliatre99atheory}; Atkey's parameterized
monad~\cite{atkey09}, which has been used to
encode disciplines like regions~\cite{kiselyov2008lightweight} and
session types~\cite{pucella2008haskell}; Devriese and
Piessens'~\cite{devriese2011information} monad-like encodings for
information flow controls; 
and many others. Oftentimes these extensions are needed to prove
stronger properties about computations, for instance to prove the
absence of information leaks or memory errors.

Unfortunately, these extensions do not enjoy the same status as monads
in terms of language support. For example, the conveniences that
Haskell provides for monadic programs (e.g., the \sfont{do} notation
combined with type-class inference) do not apply to these extensions.
One might imagine adding specialized support for each of these
extensions on a case-by-case basis, but a unifying construction into
which all of them, including normal monads, fit is clearly preferable.






This paper proposes just such a unifying construction, making several
contributions. Our first contribution is the 
definition of a \emph{polymonad}, a new way to structure effectful
computations. Polymonads give the familiar monadic bind (having type
\ls!forall $a,b$. M $a$ -> ($a$ -> M $b$) -> M $b$!) the more general
type \ls!forall $a,b$. L $a$ -> ($a$ -> M $b$) -> N $b$!. That is, a
polymonadic bind can compose computations with three different types
to a monadic bind's one. Section~\ref{sec:polymonads-alt} defines
polymonads formally, along with the \emph{polymonad laws}, which we
prove are a generalization of the monad and morphism laws. To precisely
characterize their expressiveness, we prove that polymonads correspond
to Tate's \emph{productoids}~\cite{tate12productors} (Theorem~\ref{thm:productoid}), a recent
semantic model general enough to capture most known effect systems,
including all the constructions listed above.\footnote{We
  discovered the same model concurrently with Tate and independently
  of him, though we have additionally developed supporting algorithms
  for (principal) type inference, (provably coherent) elaboration, and
(generality-preserving) simplification. Nevertheless, our presentation here has benefited from
  conversations with him.}

Whereas Tate's interest is in semantically modeling sequential
compositions of effectful computations, our interest is in supporting
practical programming in a higher-order language. Our second
contribution is the definition of \lang{}
(Section~\ref{sec:programming}), an ML-like programming language
well-suited to programming with polymonads. We work out several
examples in \lang, including novel polymonadic constructions for
stateful information flow tracking, contextual type and effect
systems~\cite{neamtiu08context}, and session types.

Our examples are made practical by \lang's support for type inference
and elaboration, which allows programs to be written in a familiar
ML-like notation while making no mention of the bind
operators. Enabling this feature, our third contribution
(Section~\ref{sec:syntactic}) is an instantiation of Jones' theory of
qualified types~\cite{jones1992theory} to \lang. In a manner similar
to Haskell's type class inference, we show that type inference for
\lang{} computes \emph{principal types} (Theorem~\ref{thm:oml}).
Our inference algorithm is equipped with an elaboration phase, which
translates source terms by inserting binds where needed.
We prove that elaboration is \emph{coherent}
(Theorem~\ref{thm:coherence}), meaning that when inference produces
constraints that could have several solutions, when these solutions
are applied to the elaborated terms the results will have equivalent
semantics, thanks to the polymonad laws. This property allows us to do
better than Haskell, which does not take such laws into account, and
so needlessly rejects programs it thinks might be ambiguous. Moreover,
as we show in Section~\ref{sec:solve}, the polymonad laws allow us to
dramatically simplify types, making them far easier to read without
compromising their generality. A prototype implementation of \lang{}
is available from the first author's web page and has been used to
check all the examples in the paper.

Put together, our work lays the foundation for providing practical
support for advanced monadic programming idioms in typed, functional
languages.

\section{Polymonads}
\label{sec:polymonads-alt}

We begin by defining polymonads formally. 
We prove that a polymonad
generalizes a collection of monads and morphisms among those
monads. We also establish a correspondence between polymonads and
productoids, placing our work on a semantic foundation that is known
to be extremely general.


\begin{definition}
\label{def:polymonad}
A \textbf{polymonad} $(\mconstrs, \Sigma)$ consists of (1) a
collection $\mconstrs$ of unary type constructors, with a
distinguished element $\mId \in \mconstrs$, such that
$\kw{Id}~\tau=\tau$, and (2) a collection, $\Sigma$, of $\mybind$
operators such that the laws below hold, 
where $\bind{(M,N)}{P} \triangleq$\ls!forall a b. M a -> (a -> N b) -> P b!.
\\[-1ex]
\begin{small}
\[\begin{array}{ll}
\multicolumn{2}{l}{$For all$~\kw{M}, \kw{N}, \kw{P}, \kw{Q}, \kw{R}, \kw{S}, \kw{T}, \kw{U} \in \mconstrs.} \\
$\textbf{(Functor)}$ & \exists \kw{b}. \kw{b}@\bind{(M,\mId)}{M} \in \Sigma ~$and$~\kw{b}~\mbox{\ls!m!}~(\lambda \kw{y}.\kw{y}) = \mbox{\ls!m!} \\[1ex]

$\textbf{(Paired morphisms)}$ & \exists \kw{b}_1@\bind{(M,\mId)}{N} \in \Sigma \iff \exists \kw{b}_2@\bind{(\mId, M)}{N} \in \Sigma~\mbox{\emph{and}} \\
                              & \forall \kw{b}_1@\bind{(M,\mId)}{N}, \kw{b}_2@\bind{(\mId, M)}{N}. \kw{b}_1\, \mbox{\ls!(f v)!}~(\lambda \kw{y}.\kw{y}) = \kw{b}_{2}~\mbox{\ls!v f!} \\[1ex]

$\textbf{(Diamond)}$       &  \exists \kw{P},\kw{b}_1,\kw{b}_2. \aset{\kw{b}_1@\bind{(M,N)}{P}, \kw{b}_2@\bind{(P,R)}{T}} \subseteq \Sigma \; \iff \\
                           & \exists \kw{S},\kw{b}_3,\kw{b}_4. \aset{\kw{b}_3@\bind{(N,R)}{S}, \kw{b}_4@\bind{(M,S)}{T}} \subseteq \Sigma \\[1ex]

$\textbf{(Associativity)}$ & \forall \kw{b}_1,\kw{b}_2,\kw{b}_3,\kw{b}_4. $If$~\\
                               & \aset{\kw{b}_1@\bind{(M,N)}{P},  \kw{b}_2@\bind{(P,R)}{T}, \kw{b}_3@\bind{(N,R)}{S}, \kw{b}_4@\bind{(M,S)}{T}}\subseteq\Sigma\\
                               & $then$~\kw{b}_2~(\kw{b}_1~m~f)~g = \kw{b}_4~m~(\lambda x. \kw{b}_3~(f~x)~g)  \\[1ex]

$\textbf{(Closure)}$           & $If$~\exists \kw{b}_1,\kw{b}_2,\kw{b}_3,\kw{b}_4. \\
                               & \aset{\kw{b}_1@\bind{(\kw{M}, \kw{N})}{\kw{P}}, 
                                       \kw{b}_2@\bind{\kw{(S,Id)}}{\kw{M}}, 
                                       \kw{b}_3@\bind{\kw{(T,Id)}}{\kw{N}},
                                       \kw{b}_4@\bind{\kw{(P,Id)}}{\kw{U}}} \subseteq \Sigma \\
                              & $then$~\exists \kw{b}. \kw{b}@\bind{(\kw{S}, \kw{T})}{\kw{U}} \in \Sigma
\end{array}\]
\end{small}
\end{definition}

Definition~\ref{def:polymonad} may look a little austere, but there is
a simple refactoring that recovers the structure of functors and monad
morphisms from a polymonad.\footnote{An online version of this paper
  provides an equivalent formulation of Definition~\ref{def:polymonad}
  in terms of join operators instead of binds. It can be found here:
  \url{http://research.microsoft.com/en-us/um/people/nswamy/papers/polymonads.pdf}. 
  The join-based definition is perhaps more natural for a
  reader with some familiarity with category theory; the bind-based
  version shown here is perhaps more familiar for a functional
  programmer.} Given $(\mathcal{M},\Sigma)$, we can easily
construct the following sets:

\begin{small}
\[\begin{array}{llcl}
$(Maps)$  & M & = & \aset{(\lambda f m. \mybind~m~f)\colon \kw{(a -> b) -> M a -> M b} \mid \mybind\colon\bind{(\mconst,\mId)}{\mconst} \in \Sigma}\\
$(Units)$ & U & = & \aset{(\lambda x. \mybind~x~(\lambda y.y))\colon \kw{a -> M a} \mid \mybind\colon\bind{(\mId,\mId)}{M} \in \Sigma}\\
$(Lifts)$ & L & = & \aset{(\lambda x. \mybind~x~(\lambda y.y))\colon \kw{M a -> N a} \mid \mybind\colon\bind{(\mconst,\mId)}{\mconst[N]} \in \Sigma}\\
\end{array}\]
\end{small}

It is fairly easy to show that the above structure satisfies
generalizations of the familiar laws for monads and monad
morphisms. For example, one can prove $\mybind~(\myunit~e)~ f = f~e$,
and $\mylift~(\myunit_1~e) = \munit_2~e$ for all suitably typed
$\myunit_1,\myunit_2 \in U$, $\mylift \in L$ and $\mybind \in
\Sigma$. 

With these intuitions in mind, one can see that the \textbf{Functor}
law ensures that each $\mconst \in \Sigma$ has a \ls$map$ in $M$, as
expected for monads.
From the construction of $L$, one can see that a bind
$\bind{(M,\mId)}{N}$ is just a morphism from $\mconst$ to
\ls$N$. Since this comes up quite often, we write
$\morph{\mconst}{\kw{N}}$ as a shorthand for $\bind{(M,\mId)}{N}$.
The \textbf{Paired morphisms} law amounts to a coherence condition
that all morphisms can be re-expressed as binds.
The \textbf{Associativity} law is the familiar associativity law
for monads generalized for both our more liberal typing for bind
operators and for the fact that we have a \emph{collection} of binds
rather than a single bind. The \textbf{Diamond} law
essentially guarantees a coherence property for associativity, namely
that it is always possible to complete an application of 
\textbf{Associativity}.
The \textbf{Closure} law ensures closure under composition of monad morphisms
with binds, also for coherence.

It is easy to prove that every collection of monads and monad
morphisms is also a polymonad. In fact, in
Appendix~\ref{sec:productoids}, we prove a stronger result
that relates polymonads to Tate's \emph{productoids}~\cite{tate12productors}.

\begin{theorem}
\label{thm:productoid}
Every polymonad gives rise to a productoid, and every productoid that
contains an \ls$Id$ element and whose joins are closed with respect to
the lifts, is a polymonad.
\end{theorem}
Tate developed productoids as a categorical
foundation for effectful computation. He
demonstrates the expressive power of productoids by showing how they
subsume other proposed extensions to
monads~\cite{wadler:2003,filinski1999representing,atkey09}.  This
theorem shows polymonads can be soundly interpreted using
productoids. Strictly speaking, productoids are more expressive than
polymonads, since they do not, in general, need to have an \ls$Id$
element, and only satisfy a slightly weaker form of our
\textbf{Closure} condition. However, these restrictions are mild, and
certainly in categories that are Cartesian closed, these conditions
are trivially met for all productoids. Thus, for programming purposes,
polymonads and productoids have exactly the same expressive power.
The development of the rest of this paper shows, for the first time,
how to harness this expressive power in a higher-order programming
language, tackling the problem of type inference, elaborating a
program while inserting binds, and proving elaboration coherent.

\section{Programming with polymonads}
\label{sec:programming}

\begin{figure}[t]
  \[\begin{array}{lllcl}
$\textit{Signatures}$ (\mconstrs,\Sigma):  
&k$-ary constructors$ & \mconstrs & ::= & \cdot \mid M/k, \mconstrs\\
&$ground constructor$ & \gm & ::= & M~\overline{\ty} \\
&$bind set$           & \Sigma & ::= & \cdot \mid \sfont{b}@s, \Sigma \\
&$bind specifications$         & s         & ::= & \forall\bar{a}. \Phi \Rightarrow \bind{(\gm_1,\gm_2)}{\gm_3} \\
&$theory constraints $ & \Phi & \\[2ex]
  $\textit{Terms:}$ & $values$          &  v       & ::=  & x \mid c \mid \slam{x}{e} \\
  & $expressions$     &  e       & ::=  & v \mid \sapp{e_1}{e_2} \mid \slet{x}{e_1}{e_2} \\
  & &          & \mid & \sif{e}{e_1}{e_2} \mid \sletrec{f}{v}{e} \\[2ex]
%
$\textit{Types:}$ &   $monadic types$   & m      & ::= & \gm \mid  \mvar\\
  & $value types$           & \ty    & ::= & a \mid T\, \overline{\ty} \mid \tfun{\ty_1}{\tapp{m}{\ty_2}} \\
  & $type schemes$     & \sigma   & ::= & \forall \bar{a}\bar\mvar. \Binds => \ty \\
  & $bag of binds$   & \Binds  & ::= & \cdot \mid \pi, \Binds \\
  & $bind type$  & \pi         & ::= & \tbind{(m_1,m_2)}{m_3}
%
  \end{array}\]
  \caption{\lang: Syntax for signatures, types, and terms}
  \label{fig:lang}
\end{figure}

\newcommand{\theorysays}{\ensuremath{\vDash}} 


\newcommand\wadler{\ensuremath{W}}
\renewcommand\mconst[1][M]{\ensuremath{\text{#1}}} 

This section presents \lang, an ML-like language for programming with
polymonads. We also present several examples that provide a flavor of
programming in \lang. As such, we aim to keep our examples as simple
as possible while still showcasing the broad applicability of
polymonads. For a formal characterization of the expressiveness of
polymonads, we appeal to Theorem~\ref{thm:productoid}.

\paragraph{Polymonadic signatures.} A \lang{} \emph{polymonadic
  signature} $(\mathcal{M}, \Sigma)$ (Figure~\ref{fig:lang}) amends
Definition~\ref{def:polymonad} in two ways. Firstly, each element
$M$ of $\mathcal{M}$ may be \emph{type-indexed}---we write
$M/k$ to indicate that $M$ is a $(k+1)$-ary type
constructor (we sometimes omit $k$ for brevity). For example,
constructor $\wadler/1$ could represent an effectful computation so
that $\wadler\;\epsilon\;\ty$ characterizes computations of type $\ty$
that have effect $\epsilon$. Type indexed constructors (rather than
large enumerations of non-indexed constructors) are critical for
writing reusable code, e.g., so we can write functions 
like $\sfont{app}: \forall a,b,\varepsilon. (a \rightarrow
\wadler\;\epsilon\;b) \rightarrow a \rightarrow \wadler\;\epsilon\;b$. 
 We write $\gm$ to denote \emph{ground
  constructors}, which are monadic constructors applied to all their
type indexes; e.g., $\wadler\;\epsilon$ is ground. Secondly, a bind
set $\Sigma$ is not specified intensionally as a set, but rather
extensionally using a language of \emph{theory constraints} $\Phi$. In
particular, $\Sigma$ is a list of mappings $\sfont{b}@s$ where $s$
contains a triple $\bind{(\gm_1,\gm_2)}{\gm_3}$ along with constraints
$\Phi$, which determine how the triple's constructors may be instantiated. For
example, a mapping $\sfont{sube}: \forall \varepsilon_1,
\varepsilon_2.\, \varepsilon_1 \subseteq \varepsilon_2 \Rightarrow
\tbind{(\wadler\;\varepsilon_1, \mId) }{\wadler\,\varepsilon_2}$
specifies the set of binds involving type indexes $\varepsilon_1,
\varepsilon_2$ such that the theory constraint $\varepsilon_1
\subseteq \varepsilon_2$ is satisfied.

\lang's type system is parametric in the
choice of theory constraints $\Phi$, which allows us to encode a
variety of prior monad-like systems with \lang.
To interpret a particular set of constraints, \lang{} requires a theory
entailment relation \theorysays. Elements of this relation, written
$\Sigma \theorysays \pi \rew \sfont{b}; \theta$, state that there
exists $\sfont{b}@\forall\bar{a}. \Phi \Rightarrow
\bind{(\gm_1,\gm_2)}{\gm_3}$ in $\Sigma$ and a substitution $\theta'$
such that $\theta\pi = \theta'\bind{(\gm_1,\gm_2)}{\gm_3}$, and the
constraints $\theta'\Phi$ are satisfiable. 
Here, $\theta$ is a substitution for the free
(non-constant) variables in $\pi$, while $\theta'$ is an instantiation
of the abstracted variables in the bind specification. Thus, the
interpretation of $\Sigma$ is the following set of binds:
$\aset{\sfont{b}@\pi \mid \Sigma \theorysays \pi \rew \sfont{b};
  \cdot}$. Signature $(\mathcal{M}, \Sigma)$ is a polymonad if this
set satisfies the polymonad laws (where each ground constructor is
treated distinctly).

Our intention is that type indices are \emph{phantom}, meaning that
they are used as a type-level representation of some property of the
polymonad's current state, but a polymonadic bind's implementation
does not depend on them.  For example, we would expect that binds
treat objects of type $\wadler\,\varepsilon\,\tau$ uniformly, for all
$\varepsilon$; different values of $\varepsilon$ could statically
prevent unsafe operations like double-frees or dangling pointer
dereferences. Of course, a polymonad may include other constructors
distinct from $\wadler$ whose bind operators could have a completely
different semantics. For example, if an object has different states
that would affect the semantics of binds, or if other effectful
features like exceptions were to be modeled, the programmer can use a
different constructor $M$ for each such feature. As such, our
requirement that the type indices are phantom does not curtail
expressiveness.

\paragraph{Terms and types.}
\lang's term language is standard. \lang{} programs do not explicitly
reference binds, but are written in \emph{direct style}, with implicit
conversions between computations of type $m\;\ty$ and their
$\ty$-typed results.  Type inference determines the bind
operations to insert (or abstract) to type check a program.

To make inference feasible, we rely crucially on \lang's call-by-value
structure. Following our prior work on monadic programming for
ML~\cite{swamy11monadICFP}, we 
restrict the shape of types assignable to a \lang{} program by
separating value types $\ty$ from the types of polymonadic
computations $m~\ty$. Here, metavariable $m$ may be either a ground
constructor $\gm$ or a polymonadic type variable $\mvar$.  The co-domain of
every function is required to be a computation type $m~\ty$, although
pure functions can be typed $\ty -> \ty'$, which is a synonym for $\ty
-> \mId~\ty'$. We also include types $T~\bar\ty$ for fully applied
type constructors, e.g., $\sfont{list}~\tint$.

Programs can also be given type schemes $\sigma$ that are polymorphic
in their polymonads, e.g., $\forall a,b,\mvar.$ $(a -> \mvar\,b) -> a ->
\mvar\,b$. Here, the variable $a$ ranges over value types $\tau$,
while $\mvar$ ranges over ground constructors $\gm$. Type schemes may also
be qualified by a set $P$ of bind constraints $\pi$. For example,
$\forall \mvar. \bind{(\mvar,\mId)}{\gm} \Rightarrow (\tint ->
\mvar~\tint) -> \gm~\tint$ is the type of a function that
abstracts over a bind having shape $\bind{(\mvar,\mId)}{\gm}$.
Notice that $\pi$ triples may contain polymonadic type variables
$\mvar$ while specification triples $s \in \Sigma$ may not. Moreover,
$\Phi$ constraints never appear in $\sigma$, which is thus entirely
independent of the choice of the theory.




\subsection{Polymonadic information flow controls}
\label{sec:ist-example}

Polymonads are appealing because they can
express many interesting constructions as we now show.

\newcommand{\entSec}{\Vdash}
\newcommand\intref{\ensuremath{\mathit{intref}}}

Figure~\ref{fig:ist} presents a polymonad $\mIST$, which implements
\emph{stateful} information flow
tracking~\cite{devriese2011information,russo08lightweight,li2006encoding,crary2005monadic,abadi1999dcc}.
The idea is that some program values are secret and some are public,
and no information about the former should be learned by observing the
latter---a property called
noninterference~\cite{goguen1982security}. In the setting of $\mIST$,
we are worried about leaks via the heap.

Heap-resident storage cells are given type $\intref\;l$ where $l$ is
the secrecy label of the referenced cell. Labels $l \in
\aset{\Lo,\Hi}$ form a lattice with order $L \sqsubset H$.  A program
is acceptable if data labeled $H$ cannot flow, directly or indirectly,
to computations or storage cells labeled $L$. In our polymonad
implementation, $\Lo$ and $\Hi$ are just types $T$ (but only ever
serve as indexes), and the lattice ordering is implemented by theory
constraints $l_1 \sqsubseteq l_2$ for $l_1,l_2 \in \aset{\Lo,\Hi}$.

\newcommand{\lfont}[1]{{\tiny \mathit{#1}}}

\begin{figure}[t]
  \small
\centering
\begin{tabular}{ll}
\begin{minipage}{3.1in}
\noindent
$\nquad\begin{array}{l@{~}c@{~}ll}
\multicolumn{4}{l}{\!\!$\textit{Signature:}$} \\
\mconstrs & = & \mIST/2 \\
  \Phi    & ::=  & \multicolumn{2}{l}{l_1 \sqsubseteq l_2 \mid \Phi_1,\Phi_2}  \\
\Sigma    & = & \sfont{bId} : & \morph{\mId}{\mId}, \\
          &   & \sfont{unitIST}: & \forall p,l.\,\morph{\mId}{\mIST\;p\;l},\\
          &   & \sfont{mapIST}: &\forall p_1,l_1,p_2,l_2.\, p_2
          \sqsubseteq p_1, l_1\sqsubseteq l_2 \Rightarrow  \\
& & & \morph{\mIST\;p_1\;l_1}{\mIST\;p_2\;l_2},\\
          &   & \sfont{appIST}: & \forall p_1,l_1,p_2,l_2.\, p_2
          \sqsubseteq p_1, l_1\sqsubseteq l_2 \Rightarrow \\
& & & \bind{(\mId,\mIST\;p_1\;l_1)}{\mIST\;p_2\;l_2}, \\
          &   & \sfont{bIST}: &\forall p_1,l_1,p_2,l_2,p_3,l_3.  \\
          &   & & l_1 \sqsubseteq p_2, l_1 \sqsubseteq l_3, l_2
          \sqsubseteq l_3, \\ 
& & & p_3 \sqsubseteq p_1,  p_3 \sqsubseteq p_2
          \Rightarrow \\
& & & \bind{(\mIST\; p_1\; l_1,\mIST\; p_2\; l_2)}{\mIST\; p_3\; l_3}
\end{array}$
\end{minipage}
&
\begin{minipage}{2.6in}
\begin{tabular}{l}
$\nquad\begin{array}{ll}
\multicolumn{2}{l}{\!\!$\textit{Types and auxiliary functions:}$} \\
\tau : & ... \mid \intref~\tau \mid \Lo \mid \Hi \\
\code{read}  : & \forall l.\, \intref~l \rightarrow \mIST\; H\; l\; \tint \\
\code{write} : & \forall l.\, \intref~l \rightarrow \tint \rightarrow \mIST\; l\; L\; ()
\end{array}$\\
~\\
\noindent
$\nquad\begin{array}{ll}
\multicolumn{2}{l}{\!\!$\textit{Example program:}$} \\
\multicolumn{2}{l}{\kw{let add_interest = lam savings. lam interest.}} \\
               & \kw{let currinterest = read interest in} \\
               & \kw{if currinterest > 0 then} \\
               & \quad\kw{let currbalance = read savings in}\\
               & \quad\kw{let newbalance =}\\
               & \quad \quad \kw{currbalance + currinterest in}\\
               & \quad\kw{write savings newbalance} \\
               & \kw{else ()}
\end{array}$
\end{tabular}
\end{minipage}
\end{tabular}
  \caption{Polymonad $\mIST$, implementing stateful
    information flow control}
  \label{fig:ist}
\end{figure}

The polymonadic constructor $\mIST/2$ uses secrecy labels for its type
indexes. A computation with type $\mIST\;p\;l\;\ty$ potentially writes
\mwh{I think $\mIST$ should be italics; at least, that's what you've
  written up to this point}
to references labeled $p$ and returns a $\ty$-result that has security
label $l$; we call $p$ the \emph{write label} and $l$ the \emph{output
  label}.  Function \ls$read$ reads a storage cell, producing a
$\mIST\; H\; l\; \tint$ computation---the second type index $l$
matches that of $l$-labeled storage cell.  Function \ls$write$ writes
a storage cell, producing a $\mIST\; l\; L\; ()$ computation---the
first type index $l$ matches the label of the written-to storage
cell. $\Hi$ is the most permissive write label and so is used for the
first index of \ls$read$, while $\Lo$ is the most permissive output
label and so is used for the second index of \ls$write$.

Aside from the identity bind $\sfont{bId}$, implemented as reverse
apply, there are four kinds of binds. Unit $\sfont{unitIST}\;p\;l$
lifts a normal term into an $\mIST$ computation. Bind
$\sfont{mapIST}\;p\;l$ lifts a computation into a more permissive
context (i.e., $p_2$ and $l_2$ are at least as permissive as $l_1$ and
$l_2$), and $\sfont{appIST}\;p\;l$ does likewise, and are
implemented using $\sfont{mapIST}$ as follows: $\sfont{appIST}\;p\;l =
\lambda x. \lambda f. \sfont{mapIST}\;p\;l\; (f\;x)\;(\lambda
x.x)$. Finally, bind $\sfont{bIST}$ composes a computation
$\mIST\;p_1\;l_1~\alpha$ with a function $\alpha ->
\mIST\;p_2\;l_2~\beta$. The constraints ensure safe information flow:
$l_1 \sqsubseteq p_2$ prevents the second computation from leaking
information about its $l_1$-secure $\alpha$-typed argument into a
reference cell that is less than $l_1$-secure. Dually, the constraints
$l_1 \sqsubseteq l_3$ and $l_2 \sqsubseteq l_3$ ensure that the
$\beta$-typed result of the composed computation is at least as secure
as the results of each component. The final constraints $p_3
\sqsubseteq p_1$ and $p_3 \sqsubseteq p_2$ ensure that the write label
of the composed computation is a lower bound of the labels of each
component.

Proving $(\mathcal{M},\Sigma)$ satisfies the polymonad laws is
straightforward. The functor and paired morphism laws are
 easy to prove. The diamond law is more tedious: we
must consider all possible pairs of binds that 
compose. This reasoning involves consideration of the theory
constraints as implementing a lattice, and so would work for any
lattice of labels, not just $\Hi$ and $\Lo$. In all, there were ten
cases to consider. We prove the associativity law for the same ten
cases. This proof is straightforward as the implementation of $\mIST$
ignores the indexes: \code{read}, \code{write} and various binds are
just as in a normal state monad, while the indexes serve only to
prevent illegal flows. Finally, proving closure is relatively
straightforward---we start with each possible bind shape and then
consider correctly-shaped flows into its components; in all there were
eleven cases.

\paragraph{Example.} 
The lower right of Figure~\ref{fig:ist} shows an example use of
$\mIST$. The
$\kw{add_interest}$ function takes two reference cells, $\kw{savings}$
and $\kw{interest}$, and modifies the former by adding to it the
latter if it is non-negative.\footnote{For ease of presentation, the program 
in Figure~\ref{fig:ist} uses \ls$let$ to sequence computations. This is not 
essential, e.g., we need not have \ls$let$-bound \ls$currbalance$.} Notice that
expressions of type $\mIST\;p\;l\;\ty$ are used as if they merely had
type $\ty$---see the branch on \ls|currinterest|, for
example. The program is rewritten during type inference to insert, or
abstract, the necessary binds so that the program type checks. This
process results in the following type for
\ls$add_interest$:\footnote{This and other example types were
  generated by our prototype implementation.} 
\[\small\begin{array}{l@{~}l}
\multicolumn{2}{l}{\forall \mvar_6,\mvar_{27}, a_1, a_2. \Binds =>
 \lfont{\intref\;a_1 \rightarrow \intref\;a_2 \rightarrow
   \mvar_{27}\;()}} \\
\text{where } \Binds = &
\bind{(\mId,\mId)}{\mvar_{6}}, \bind{(\mIST\; \Hi\; a_1, \mIST\; a_1\; \Lo)}{\mvar_{6}}, \bind{(\mIST\; \Hi\; a_2,\mvar_{6})}{\mvar_{27}}
\end{array}
\]
The rewritten version of \ls$add_interest$ starts with a sequence of
$\lambda$ abstractions, one for each of the bind constraints in
$\Binds$.  If we imagine these are numbered $\sfont{b1}$
... $\sfont{b3}$, e.g., where $\sfont{b1}$ is a bind with type
$\bind{(\mId,\mId)}{\mvar_{6}}$, then the term looks as follows
(notation $\kw{...}$ denotes code elided for simplicity):
\begin{lstlisting}
     lam savings. lam interest. b3 (read interest)
       (lam currinterest. if currinterest > 0 then (b2 ...) else (b1 () (lam z. z)))
\end{lstlisting}

In a program that calls \ls|add_interest|, the bind constraints will
be solved, and actual implementations of these binds will be passed in
for each of $\sfont{b}_i$ (using a kind of dictionary-passing style as
with Haskell type classes).

Looking at the type of \ls|add_interest| we can see how the
constraints prevent improper information flows. In particular, if we
tried to call \kw{add_interest} with $a_1 = \Lo$ and $a_2 = \Hi$, then
the last two constraints become $\bind{(\mIST\; \Hi\; \Lo, \mIST\;
  \Lo\; \Lo)}{\mvar_{6}}, \bind{(\mIST\; \Hi\;
  \Hi,\mvar_{6})}{\mvar_{27}}$, and so we must instantiate $\mvar_6$
and $\mvar_{27}$ in a way allowed by the signature in
Figure~\ref{fig:ist}. While we can legally instantiate $\mvar_6 =
\mIST\;\Lo\;l_3$ for any $l_3$ to solve the second constraint, there
is then no possible instantiation of $\mvar_{27}$ that can solve the
third constraint.  After substituting for $\mvar_6$, this constraint
has the form $\tbind{(\mIST\;\Hi\; \Hi, \mIST\;\Lo\;l_3)}{\mvar_{27}}$,
but this form is unacceptable because the $\Hi$ output of the first
computation could be leaked by the $\Lo$ side effect of the second
computation. On the other hand, all other instantiations of $a_1$ and
$a_2$ (e.g., $a_1 = \Hi$ and $a_2 = \Lo$ to correspond to a secret
savings account but a public interest rate) do have solutions and do
not leak information. 
Having just discussed the latter two constraints, consider the 
first, $\bind{(\mId,\mId)}{\mvar_{6}}$. This constraint is
important because it says that $\mvar_6$ must have a unit, which is
needed to properly type the else branch; units are not required of a
polymonad in general.  

The type given above for \ls|add_interest| is not its principal type, but an
\emph{improved} one. As it turns out, the principal type is
basically unreadable, with 19 bind constraints! Fortunately,
Section~\ref{sec:solve} shows how some basic rules can greatly
simplify types without reducing their applicability, as has been done
above. Moreover, our coherence result (given in the next section)
assures that the corresponding changes to the elaborated term do not
depend on the particular simplifications: the polymonad laws ensure
all such elaborations will have the same semantics.


\subsection{Contextual type and effect systems}

\newcommand{\CE}{\ensuremath{\mbox{\textit{CE}}}}

\mwh{Here again, $\CE$ maybe should be italics (see session types in
  the next figure)}
Wadler and Thiemann~\cite{wadler:2003} showed how a monadic-style
construct can be used to model type and effect systems.  Polymonads
can model standard effect systems, but more interestingly can be used
to model \emph{contextual effects}~\cite{neamtiu08context}, which
augment traditional effects with the notion of \emph{prior}
and \emph{future} effects of an expression within a broader
context. As an example, suppose we are using a language that
partitions memory into \emph{regions} $R_1, ..., R_n$ and reads/writes
of references into region $R$ have effect $\aset{R}$. Then in the
context of the program $\kw{read}\;r_1; \kw{read}\;r_2$, where $r_1$
points into region $R_1$ and $r_2$ points into region $R_2$, the
contextual effect of the subexpression $\kw{read}\;r_1$ would be the
triple $[ \emptyset; \aset{R_1}; \aset{R_2} ]$: the prior effect is
empty, the present effect is $\aset{R_1}$, and the future effect is
$\aset{R_2}$.


\begin{figure}[t]
\hspace*{-.1in}
\begin{tabular}{ll}
\begin{minipage}{2.9in}
\noindent
$\begin{array}{lcl}
\mconstrs & = & \CE/3 \\
\Sigma  & = & \sfont{bId} : \tbind{(\mId,\mId)}{\mId}, \\
&&\sfont{unitce}:  \tbind{(\mId,\mId)} 
                {\CE\,\top\,\emptyset\,\top}\\
&& \sfont{appce}:  \forall
\alpha_1,\alpha_2,\epsilon_1,\epsilon_2,\omega_1,\omega_2. \\
&& \quad (\alpha_2 \subseteq \alpha_1, \epsilon_1 \subseteq \epsilon_2, \omega_2 \subseteq \omega_1) \Rightarrow\\
&&\quad \tbind{(\mId,\CE\;\alpha_1\;\epsilon_1\,\omega_1) 
                }{\CE\,\alpha_2\;\epsilon_2\,
                \omega_2} \\
&& \sfont{mapce}:  \forall
\alpha_1,\alpha_2,\epsilon_1,\epsilon_2,\omega_1,\omega_2. \\
&& \quad (\alpha_2 \subseteq \alpha_1, \epsilon_1 \subseteq \epsilon_2, \omega_2 \subseteq \omega_1) \Rightarrow\\
&&\quad \tbind{(\CE\;\alpha_1\;\epsilon_1\,\omega_1, \mId) 
                }{\CE\,\alpha_2\;\epsilon_2\,
                \omega_2} \\
&&\sfont{bindce}:  \forall
\alpha_1,\epsilon_1,\omega_1,\alpha_2,\epsilon_2,\omega_2,\epsilon_3. \\
&& \quad\epsilon_2 \cup \omega_2 = \omega_1, \epsilon_1\cup \alpha_1 = \alpha_2, \epsilon_1 \cup \epsilon_2 = \epsilon_3) \Rightarrow \\
&&\quad \tbind{(\CE\;\alpha_1\;\epsilon_1\,\omega_1, 
                \CE\, \alpha_2\;\epsilon_2\,\omega_2)}{\CE\,\alpha_1\;\epsilon_3\,\omega_2} 
\end{array}$
\end{minipage}
&
\begin{minipage}{3in}
\noindent
$\begin{array}{ll}
\multicolumn{2}{l}{\!\!$\textit{Types and theory constraints:}$} \\
\ty     & ::=  ... \mid \{A_1\} ... \{A_n\} \mid \emptyset \mid \top \mid \ty_1 \cup \ty_2 \\
\Phi    & ::= \ty \subseteq \ty'  \mid \ty = \ty' \mid \Phi,\Phi \\
\\
\multicolumn{2}{l}{\!\!$\textit{Auxiliary functions:}$} \\
\code{read}  : & \forall \alpha,\omega,r.\, \intref~r \rightarrow \CE\; \alpha\; r\; \omega\; \tint \\
\code{write} : & \forall \alpha,\omega,r.\, \intref~r \rightarrow \tint \rightarrow \CE\; \alpha\;r\; \omega\; ()\\
\\
\\
\\
\\
\\
\end{array}$
\end{minipage}
\end{tabular}
\caption{Polymonad expressing contextual type and effect systems}
\label{fig:ctxeff}
\end{figure}

Figure~\ref{fig:ctxeff} models contextual effects as the polymonad
$\CE~\alpha~\epsilon~\omega~\ty$, for the type of a computation with
prior, present, and future effects $\alpha$, $\epsilon$, and $\omega$,
respectively. Indices are sets of atomic effects $\{A_1\}
... \{A_n\}$, with $\emptyset$ the empty effect, $\top$ the effect set
that includes all other effects, and $\cup$ the union of two
effects. We also introduce theory constraints for subset relations and
extensional equality on sets, with the obvious interpretation. As an
example source of effects, we include \code{read} and \code{write}
functions on references into region sets $r$.  The bind $\sfont{unitce}$
ascribes a pure computation as having an empty effect and
any prior and future effects. The binds $\sfont{appce}$ and
$\sfont{mapce}$ express that it is safe to consider an additional
effect for the current computation (the $\epsilon$s are covariant),
and fewer effects for the prior and future computations ($\alpha$s and
$\omega$s are contravariant). Finally, $\sfont{bindce}$ composes two
computations such that the future effect of the first computation
includes the effect of the second one, provided that the prior effect
of the second computation includes the first computation; the effect
of the composition includes both effects, while the prior effect is
the same as before the first computation, and the future effect is the
same as after the second computation.

\subsection{Parameterized monads, and session types}

\newcommand\atkey{\ensuremath{A}}

Finally, we show $\langext$ can express Atkey's parameterized
monad~\cite{atkey09}, which has been used to
encode disciplines like regions~\cite{kiselyov2008lightweight} and
session types~\cite{pucella2008haskell}. The type constructor
$\atkey~p~q~\ty$ can be thought of (informally) as the type of a
computation producing a $\ty$-typed result, with a pre-condition $p$
and a post-condition $q$. 

\begin{figure}[t]
\begin{tabular}{ll}
\begin{minipage}{2.9in}
\noindent
$\begin{array}{lcl}
\mconstrs & = & \mId,\atkey/2 \\
\Sigma    & = & \sfont{bId} : \tbind{(\mId,\mId)}{\mId}, \\
          &   & \sfont{mapA}: \forall p,r.\; \tbind{(\atkey\;p\;r,\mId)}{\atkey\;p\;r}, \\
          &   & \sfont{appA}: \forall p,r.\; \tbind{(\mId,\atkey\;p\;r)}{\atkey\;p\;r}, \\
          &   & \sfont{unitA}: \forall p.\; \tbind{(\mId,\mId)}{\atkey\;p\;p}, \\
          &   & \sfont{bindA}:\forall p,q,r.\; \tbind{(\atkey\,p\,q,\; \atkey\,q\,r)}{\atkey\,p\,r} \\
\end{array}$
\end{minipage}
& 
\begin{minipage}{3.1in}
\begin{tabular}{l}
$\nquad\begin{array}{ll}
\multicolumn{2}{l}{\!\!$\textit{Types:}$} \\
\multicolumn{2}{l}{\ty  ::= \dots \mid \tsend{\ty_1}{\ty_2} \mid \trecv{\ty_1}{\ty_2} \mid \tdone} \\
\\
\multicolumn{2}{l}{\!\!$\textit{Auxiliary functions:}$} \\
\sfont{send}\,  : & \forall a,q.\,a\,\xrightarrow\, \atkey\, (\tsend{a}{q})\,q\, () \\
\sfont{recv}\,  : & \forall a,q.\,()\,\xrightarrow\, \atkey\, (\trecv{a}{q})\,q\, a 
\end{array}$
\end{tabular}
\end{minipage}
\end{tabular}
\caption{Parameterized monad for session types, expressed as a
  polymonad}
\label{fig:session}
\end{figure}

As a concrete example, Figure~\ref{fig:session} gives a polymonadic
expression of Pucella and Tov's notion of session
types~\cite{pucella2008haskell}.  The type $\atkey\,{p}\,{q}\, \ty$
represents a computation involved in a two-party session which starts
in protocol state $p$ and completes in state $q$, returning a value of
type $\ty$.  The key element of the signature $\Sigma$ is the
$\sfont{bindA}$, which permits composing two computations where the
first's post-condition matches the second's precondition.  We use the
type index $\tsend{\ty}{q}$ to denote a protocol state that requires a
message of type $\ty$ to be sent, and then transitions to
$q$. Similarly, the type index $\trecv{\ty}{r}$ denotes the protocol
state in which once a message of type $\ty$ is received, the protocol
transitions to $r$. We also use the index $\tdone$ to denote the
protocol end state. The signatures of two primitive operations for
sending and receiving messages capture this behavior.

As an example, the following \langext{} program
implements one side of a simple protocol that sends a message
\ls$x$, waits for an integer reply \ls$y$, and returns \ls$y+1$.
\[
\begin{array}{c}
\kw{let go = lam x. let _ = send x in incr (recv ())} \\
$Simplified type: $\forall a,b,q,\mvar.\,
\tbind{(\atkey\,(\tsend\,a\,b)\, b,\;\atkey\,(\trecv\,\tint\;q)\, q))}{\mvar}
\Rightarrow \,(a\rightarrow\mvar\;\tint) \\
\end{array}
\]
There are no specific theory constraints for session types:
constraints simply arise by unification and are solved as usual when
instantiating the final program (e.g., to call \ls$go 0$).
\section{Coherent type inference for \lang}
\label{sec:syntactic}
\renewcommand\mconst[1][M]{\ensuremath{\sfont{#1}}} 
\newcommand{\evidence}[1]{\mathsf{app}(#1)}
\newcommand{\evlift}[2]{\mathsf{b}_{#1,#2}\,}
\newcommand{\evbind}[3]{\mathsf{b}_{#1,#2,#3}\,}

\begin{figure*}[tH!]
{\begin{small}\[
\begin{array}{l}
  \fbox{$\Binds |= \Binds'$} \qquad
  \inference{\forall \pi \in \Binds'. \pi \in \Binds \vee \pi \in \Sigma}
              {\Binds |= \Binds'}[(TS-Entail)]
\\\\
  \fbox{$\Binds |= \sigma \tygt \ty\;\leadsto\mathsf{f}$} \qquad
  \inference{
    \theta = [\bar \tau/\bar{a}][\bar{m}/\bar{\mvar}] & \Binds |= \theta\Binds_1}
            {\Binds |= (\tscheme{\bar{a}\bar{\mvar}}{\Binds_1}{\ty}) \,\tygt\,
              {\theta\ty} \;\leadsto \evidence{\theta\Binds_1}}[(TS-Inst)]
\\\\
  \fbox{$\prefix\Binds\Gamma v : \ty \;\leadsto\mathsf{e}$} \qquad

  \inference{v\in\aset{x,c} & \Binds |= \Gamma(v) \tygt \ty \;\leadsto \mathsf{f}}
            {\prefix\Binds\Gamma v : \ty \;\leadsto \mathsf{f}\,v}[(TS-XC)]
\\\\
  \inference{\prefix\Binds{\Gamma,x@\ty_1} e : \tapp{m}{\ty_2} \;\leadsto \mathsf{e}}
            {\prefix\Binds\Gamma \slam{x}{e} : \tfun{\ty_1}{\tapp{m}{\ty_2}} \;\leadsto \slam{x}{\mathsf{e}}}[(TS-Lam)]
\\\\
  \fbox{$\prefix\Binds\Gamma e : \tapp{m}\ty \;\leadsto\mathsf{e}$}\quad 

  \inference{\prefix\Binds\Gamma v : \ty \;\leadsto \mathsf{e}}
            {\prefix{\Binds,\morph{\mconst[Id]}{m}}\Gamma v : \tapp{m}{\ty} 
            \;\leadsto \evbind{\mconst[Id]}{\mconst[Id]}{m}{\mathsf{e}}\;(\lambda x.x)}[(TS-V)]
\\\\
  \inference{\prefix{\Binds_1}{\Gamma,x@\ty} v : \ty \;\leadsto \mathsf{e}_1 & 
    (\sigma,\mathsf{e_2}) = \Gen{\Gamma}{\Binds_1 => \ty,\,\mathsf{e_1}} \\
    \prefix\Binds{\Gamma,x@\sigma} e : \tapp{m}{\ty'} \;\leadsto \mathsf{e}_3}
            {\prefix\Binds\Gamma \sletrec{x}{v}{e} : \tapp{m}{\ty'} 
            \;\leadsto \sletrec{x}{\mathsf{e}_2}{\mathsf{e}_3}
            }[(TS-Rec)]
\\\\
  \inference{\prefix{\Binds_1}\Gamma v : \ty \;\leadsto \mathsf{e}_1 &
    (\sigma,\mathsf{e}_2) = \Gen{\Gamma}{\Binds_1 => \ty,\,\mathsf{e}_1} \\
    \prefix\Binds{\Gamma,x@\sigma} e : \tapp{m}\ty' \;\leadsto\mathsf{e}_3}
            {\prefix\Binds\Gamma \slet{x}{v}{e}  : \tapp{m}\ty'
            \;\leadsto \slet{x}{\mathsf{e}_2}{\mathsf{e}_3}
            }[(TS-Let)]
\\\\
  \inference{ \prefix\Binds\Gamma e_1 : \tapp{m_1}{\ty_1} \;\leadsto \mathsf{e}_1 &
    \prefix\Binds{\Gamma,x@\ty_1} e_2 : \tapp{m_2}{\ty_2}  \;\leadsto \mathsf{e}_2  \\
    e_1 \neq v &\Binds |= (m_1,m_2) \rhd {m_3}}
            {\prefix\Binds\Gamma \slet{x}{e_1}{e_2} : \tapp{m_3}{\ty_2}
              \;\leadsto \evbind{m_1}{m_2}{m_3}{\mathsf{e}_1}\,{(\lambda x.\,\mathsf{e}_2)}
            }[(TS-Do)]
\\\\
  \inference{\prefix\Binds\Gamma  e_1 : \tapp{m_1}{(\tfun{\ty_2}{\tapp{m_3}{\ty}})}  \;\leadsto \mathsf{e}_1 &
    \prefix\Binds\Gamma  e_2 : \tapp{m_2}{\ty_2}  \;\leadsto \mathsf{e}_2 \\
    \Binds |= {(m_1,m_4)}\rhd{m_5} &
    \Binds |= {(m_2,m_3)}\rhd{m_4}  }
            {\prefix\Binds\Gamma  \eapp{e_1}{e_2} : \tapp{m_5}{\ty} 
             \;\leadsto \evbind{m_1}{m_4}{m_5}{\mathsf{e}_1}\;{(\evbind{m_2}{m_3}{m_4}{\mathsf{e}_2}})}[(TS-App)]
\\\\
  \inference{\prefix\Binds\Gamma e_1 : \tapp{m_1}\kw{bool} \;\leadsto \mathsf{e}_1 &
    \prefix\Binds\Gamma e_2 : \tapp{m_2}{\ty} \;\leadsto \mathsf{e}_2  \\
    \prefix\Binds\Gamma e_3 : \tapp{m_3}{\ty} \;\leadsto \mathsf{e}_3 &
    \Binds |= \morph{m_2}{m},  \morph{m_3}{m}, {(m_1,m)}\rhd{m'}}
      {\prefix\Binds\Gamma \sif{e_1}{e_2}{e_3} : \tapp{m'}{\ty}
      }[(TS-If)]
       \\{ \;\leadsto \evbind{m_1}{m}{m'}{\mathsf{e}_1}\,{(\lambda b.\,\mathsf{if}\;b\;\mathsf{then}\;\evbind{m_2}{\mconst[Id]}{m}{\mathsf{e}_2}\;(\lambda x.x)\;\mathsf{else}\; \evbind{m_3}{\mconst[Id]}{m}{\mathsf{e}_3}\; (\lambda x. x))}
        }
\iffull
\else
\\\\
\begin{array}{ll}
  \Gen{\Gamma}{\Binds => \ty, \mathsf{e}} 
    & = (\forall (\ftv{\Binds => \ty} \setminus \ftv{\Gamma}). \Binds => \ty,\;\mathsf{abs}(\Binds,\mathsf{e}))\\
  \mathsf{abs}((\tbind{(m_1,m_2)}{m_3},P),\mathsf{e}) &= \lambda\evbind{m_1}{m_2}{m_3}.\,\mathsf{abs}(P,\mathsf{e})\\
  \mathsf{abs}(\cdot,\mathsf{e})                    &= \mathsf{e} \\      
  \evidence{P,\tbind{(m_1,m_2)}{m_3})} &= \lambda f.\,\evidence{P}\,(f\;\evbind{m_1}{m_2}{m_3})\\
  \evidence{\cdot} &= \lambda x.\,x
\end{array}
\fi
  \end{array}\]\end{small}}
  \caption{Syntax-directed type rules for \lang{}, where $\Sigma$
    is an implicit parameter. }
  \label{fig:ssyntaxrules}
\end{figure*}

This section defines our declarative type system for \lang{} and
proves that type inference produces principal types, and that
elaborated programs are coherent.

Figure~\ref{fig:ssyntaxrules} gives
a syntax-directed type system,
organized into two main judgments. The value-typing judgment
$\prefix{\Binds}{\Gamma} v : \ty \;\leadsto\mathsf{e}$ types a value $v$ in an environment
$\Gamma$ (binding variables $x$ and constants $c$ to type schemes) at
the type $\ty$, provided the constraints $\Binds$ are satisfiable. 
Moreover, it \emph{elaborates} the value $v$ into a lambda term $\mathsf{e}$
that explicitly contains binds, lifts, and evidence passing (as shown
in Section~\ref{sec:ist-example}). However, note that the elaboration is
independent and we can read just the typing rules by igoring the elaborated terms.
The expression-typing judgment $\prefix{\Binds}{\Gamma} e :
\tapp{m}{\ty}\;\leadsto\mathsf{e}$ 
is similar, except that it yields a computation type. Constraint
satisfiability $\Binds |= \Binds'$, defined in the figure, states that
$\Binds'$ is satisfiable under the hypothesis $\Binds$ if $\Binds'
\subseteq \Binds \cup \Sigma$ where we consider
$\pi \in \Sigma$ if and only if $\Sigma 
\theorysays \pi \rew \sfont{b}; \cdot$ (for
some $\sfont{b}$).

The rule (TS-XC) types a variable or constant at an instance of its
type scheme in the environment. The instance relation for type schemes
$\Binds |= \sigma \geq \ty\;\leadsto\mathsf{f}$ is standard: it instantiates the
bound variables, and checks that the abstracted constraints are
entailed by the hypothesis $\Binds$. The elaborated $\mathsf{f}$ term
supplies the instantiated evidence using the \textsf{app} form. The rule (TS-Lam) is
straightforward where the bound variable is given a value
type and the body a computation type.

The rule (TS-V) allows a value $v:\ty$ to be used as an expression
by lifting it to a computation type $\tapp{m}{\ty}$, so long
as there exists a morphism (or unit) from the \ls$Id$ functor to
$m$. The elaborated term uses $\evbind{\mconst[Id]}{\mconst[Id]}{m}$ to lift
explicitly to monad $m$. Note that for evidence we make up names
($\evbind{\mconst[Id]}{\mconst[Id]}{m}$)  based on the constraint ($\morph{\mconst[Id]}{m}$).
This simplifies our presentation but an implementation would
name each constraint explicitly \cite{jones1994improvement}.
We use the name $\evbind{m_1}{\mconst[Id]}{m_2}$ for morphism constraints $\morph{m_1}{m_2}$,
and use $\evbind{m_1}{m_2}{m_3}$ for general bind constraints ${(m_1,m_2)}\rhd{m_3}$.

(TS-Rec) types a recursive let-binding by typing the definition $v$ at
the same (mono-)type as the \ls$letrec$-bound variable $f$. When
typing the body $e$, we generalize the type of $f$ using a standard
generalization function $\Gen{\Gamma}{\Binds => \ty,\;\mathsf{e}}$, which closes
the type relative to $\Gamma$ by generalizing over its free type
variables. However, in constrast to regular generalization, we return
both a generalized type, as well as an elaboration of $\mathsf{e}$ that
takes all generalized constraints as explicit evidence parameters (as
defined by rule $\mathsf{abs}$).
(TS-Let) is similar, although somewhat simpler since there is no
recursion involved.

(TS-Do) is best understood by looking at its elaboration: since we are in a call-by-value setting,
we interpret a \ls$let$-binding as forcing and sequencing two
computations using a single bind where $e_1$ is typed monomorphically.

(TS-App) is similar to (TS-Do), where, again, since we use
call-by-value, in the elaboration we sequence the function and its
argument using two bind operators, and then apply the function.
(TS-If) is also similar, since we sequence the expression $e$ in the
guard with the branches. As usual, we require the
branches to have the same type. This is achieved by generating
morphism constraints, $\morph{m_2}{m}$ and $\morph{m_3}{m}$ to coerce
the type of each branch to a functor $m$ before sequencing it
with the guard expression.

\subsection{Principal types}

\newcommand{\el}[1]{\llbracket #1\rrbracket}
\newcommand{\retTrans}{\textsf{ret}}
\newcommand{\bindTrans}{\textsf{do}}
\newcommand{\appTrans}{\textsf{app}}
\newcommand{\ifTrans}{\textsf{cond}}
\newcommand{\recTrans}{\textsf{rec}}
\newcommand{\prefixOML}[2]{\prefix{#1}{#2}_{\textsc{\tiny OML}}}

\begin{figure}[t]
  \[
  \begin{array}{ll}
    \el{x}^\kstar  &= x \\
    \el{c}^\kstar  &= c \\
    \el{\lambda x.e}^\kstar &= \lambda x. \el{e} \\
    \\
    \el{v} &= \mathtt{ret}\;\el{v}^\kstar \\
    \el{e_1\; e_2} &= \mathtt{app}\; \el{e_1}\; \el{e_2} \\
    \el{\slet{x}{v}{e}} &= \slet{x}{\el{v}^\kstar}{\el{e}} \\
    \el{\slet{x}{e_1}{e_2}} &= \mathtt{do}\;\el{e_1}\; \el{\slam{x}{e_2}}^\kstar \qquad\textrm{(when $e_1 \neq v$)}\\
    \el{\sif{e_1}{e_2}{e_3}} &= \mathtt{cond}\; \el{e_1}\; \slam{()}{\el{e_2}} \; \slam{()}{\el{e_3}} \\
    \el{\sletrec{f}{v}{e}} &= \mathtt{letrec}\; {f = \el{v}^\kstar}~\mathtt{in}~{\el{e}}
  \end{array}
  \]

  \[
  \begin{array}{ll}
    \retTrans  &: \forall\alpha \mvar.\,(\morph{\mId}{\mvar}) => \tfun{\alpha}{\tapp{\mvar}{\alpha}} \\

    \bindTrans &: \forall\alpha \beta \mvar_1 \mvar_2 \mvar.\,(\bind{(\mvar_1,\mvar_2)}{\mvar}) => \tfun{\tapp{\mvar_1}{\alpha}}{(\tfun{\tfun{\alpha}{\tapp{\mvar_2}{\beta}})}{\tapp{\mvar}{\beta}}} \\

    \appTrans &: \forall\alpha \beta \mvar_1 \mvar_2 \mvar_3 \mvar_4 \mvar.\, (\bind{(\mvar_1,\mvar_4)}{\mvar}, \bind{(\mvar_2,\mvar_3)}{\mvar_4}) => \tfun{\tapp{\mvar_1}{(\tfun{\alpha}{\tapp{\mvar_3}{\beta}})}}{\tfun{\tapp{\mvar_2}{\alpha}}{\tapp{\mvar}{\beta}}} \\

    \ifTrans &: \forall\alpha\mvar_1\mvar_2\mvar_3\mvar\mvar'.\, (\morph{\mvar_2}{\mvar}, \morph{\mvar_3}{\mvar}, \bind{(\mvar_1,\mvar)}{\mvar'}) \\
    & \qquad=> \tfun{\tapp{\mvar_1}{\kw{bool}}}
             {\tfun{(\tfun{()}{\tapp{\mvar_2}{\alpha}})}
               {\tfun{(\tfun{()}{\tapp{\mvar_3}{\alpha}})}
                 {\tapp{\mvar'}{\alpha}}}} 
  \end{array}\]
  \caption{Type inference for \lang{} via elaboration to OML}
  \label{fig:xlate-oml}
\end{figure}

The type rules admit principal types, and there exists an efficient
type inference algorithm that finds such types. The way we 
show this is by a translation
of polymonadic terms (and types) to terms (and types) in Overloaded ML
(OML)~\cite{jones1992theory} and prove this translation is sound and 
complete: a polymonadic term is well-typed if and only if its
translated OML term has an equivalent type. 
OML's type inference algorithm is
known to enjoy principal types, so a corollary of our translation is
that principal types exist for our system too.  

We encode terms in our language into OML as shown in
Figure~\ref{fig:xlate-oml}.  We rely on four primitive OML terms that
force the typing of the terms to generate the same constraints as our
type system does: $\retTrans$ for lifting a pure term, 
$\bindTrans$ for typing a do-binding, 
$\appTrans$ for typing an application, 
and $\ifTrans$ for conditionals.  Using these
primitives, we encode values and expressions of our system into OML.

We write $\prefixOML\Binds\Gamma e : \ty$ for a derivation in the
syntax directed inference system of OML (cf. Jones~\cite{jones1992theory},
Fig. 4).  

\begin{theorem}[Encoding to OML is sound and complete]
  \label{thm:oml}
  \strut\\\textbf{Soundness}: Whenever $\prefix{\Binds}{\Gamma} v : \tau$
  we have $\prefixOML{\Binds}{\Gamma} \el{v}^\kstar : \tau$.
  Similarly, whenever $\prefix\Binds\Gamma e : \tapp{m}{\ty}$
  then we have $\prefixOML\Binds\Gamma \el{e} : \tapp{m}{\ty}$.

  \noindent\textbf{Completeness}: Whenever
  $\prefixOML\Binds\Gamma \el{v}^\kstar : \tau$, then we have
  $\prefix\Binds\Gamma v : \tau$. Similarly, whenever
  $\prefixOML\Binds\Gamma \el{e} : \tapp{m}{\ty}$, then we have
  $\prefix\Binds\Gamma e : \tapp{m}{\ty}$.
\end{theorem}

\noindent The proof is by straightforward induction on the typing derivation of
the term. It is important to note that our system uses the same
instantiation and generalization relations as OML which is required
for the induction argument. Moreover, the constraint entailment over
bind constraints also satisfies the monotonicity, transitivity and
closure under substitution properties required by OML.
As a corollary of the above properties, our system admits
principal types via the general-purpose OML type inference algorithm.

\subsection{Ambiguity}

Seeing the previous OML translation, one might think we could
directly translate our programs into Haskell since Haskell uses
OML style type inference.
Unfortunately, in practice, Haskell would
reject many useful programs. In particular, Haskell
rejects as ambiguous any term whose type $\forall
\bar{\alpha}. \Binds => \ty$ includes a variable $\alpha$ that 
occurs free in $\Binds$
but not in $\ty$;\footnote{The
actual ambiguity rule in Haskell is more involved due to functional dependencies
and type families but that does not affect our results.} we
call such type variables \emph{open}.
Haskell, in its generality, must reject such terms since the
instantiation of an open variable can have operational effect,
while at the same time, since the variable does not appear in $\ty$, the
instantiation for it can never be uniquely determined by the context
in which the term is used. A common example is the term 
\lstinline|show . read| with the type \lstinline|(Show a, Read a) => String -> String|,
where \lstinline|a| is open. Depending on the instantiation of \lstinline|a|,
the term may parse and show integers, or doubles, etc.

Rejecting all types that contain open variables works well for type
classes, but it would be unacceptable for \lang. Many simple terms
have principal types with open variables. For example, the term
$\slam{f}{\slam{x}{\sapp{f}{x}}}$ has type $\forall a b \mvar_1
\mvar_2 \mvar_3.$ $((\mathsf{Id},\mvar_1)\rhd\;\mvar_2,
(\mathsf{Id},\mvar_2)\rhd\;\mvar_3)$ $\Rightarrow\;(a \rightarrow
\mvar_1\;b) \rightarrow \alpha \rightarrow \mvar_3\;b$ where type
variable $\mvar_2$ is open. 

In the special case where there is only one polymonadic constructor
available when typing the program, the coherence problem is moot,
e.g., say, if the whole program were to only be typed using only the
\mIST{} polymonad of Section~\ref{sec:ist-example}. However, recall
that polymonads generalize monads and morphisms, for which there can
be coherence issues (as is well known), so polymonads must address
them. As an example, imagine combining our $\mIST$ polymonad (which
generalizes the state monad) with an exception monad $\sfont{Exn}$,
resulting in an $\sfont{ISTExn}$ polymonad.  Then, an improperly coded
bind that composed $\mIST$ with $\sfont{Exn}$ could sometimes reset
the heap, and sometimes not (a similar example is provided by
Filinski~\cite{filinski94representing}).

A major contribution of this paper is that for binds that satisfy the
polymonad laws, we need not reject all types with open
variables. In particular, by appealing to the polymonadic laws, we can
prove that programs with open type variables in bind constraints are
indeed unambiguous. Even if there are many possible instantiations,
the semantics of each instantiation is equivalent, enabling us to
solve polymonadic constraints much more aggressively. This
coherence result is at the essence of making programming with
polymonads practical.

\subsection{Coherence}
\label{sec:coherence}

\mwh{The key to point out here is that you can ignore the type indexes
  since they have no impact on computation, and thus pretend that all
  constructors are unary. Then the argument is as before.}

The main result of this section (Theorem~\ref{thm:coherence})
establishes that for a certain class of polymonads, the ambiguity
check of OML can be weakened to accept more programs while still
ensuring that programs are coherent.  Thus, for this class of
polymonads, programmers can reliably view our syntax-directed system
as a specification without being concerned with the details of how the
type inference algorithm is implemented or how programs are
elaborated.

The proof of Theorem~\ref{thm:coherence} is a little technical---the
following roadmap summarizes the structure of the development.

\newcommand\unamb[3]{\ensuremath{\mathsf{unambiguous}(#1,#2,#3)}}

\begin{itemize}
\item We define the class of \emph{principal} polymonads for which
  unambiguous typing derivations are coherent. All polymonads that we
  know of are principal.

\item Given $\prefix{\Binds}{\Gamma} e : t \rew \tgte$ (with $t \in
  \aset{\ty, \tapp{m}{\ty}}$), the predicate
  $\unamb{\Binds}{\Gamma}{t}$ characterizes when the derivation is
  unambiguous. This notion requires interpreting $\Binds$ as a graph
  $G_\Binds$, and ensuring (roughly) that all open variables in
  $\Binds$ have non-zero in/out-degree in $G_\Binds$.

\item A \emph{solution} $S$ to a constraint graph with respect to a
  polymonad $(\mconstrs, \Sigma)$ is an assignment of ground polymonad
  constructors $\mconst\in\mconstrs$ to the variables in the graph such that
  each instantiated constraint is present in $\Sigma$. We give an
  equivalence relation on solutions such that $S_1 \cong S_2$ if they
  differ only on the assignment to open variables in a manner where
  the composition of binds still computes the same function according
  to the polymonad laws.

\item Finally, given $\prefix{\Binds}{\Gamma} e : t \rew \tgte$
  and $\unamb{\Binds}{\Gamma}{t}$, we prove that all solutions
  to $\Binds$ that agree on the free variables of $\Gamma$ and $t$ are
  in the same equivalence class.
\end{itemize}

While Theorem~\ref{thm:coherence} enables our type system to be used
in practice, this result is not the most powerful theorem one can
imagine. Ideally, one might like a theorem of the form
$\prefix{\Binds}{\Gamma} e : t \rew \tgte$ and
$\prefix{\Binds'}{\Gamma} e : t \rew \tgte'$ implies $\tgte$ is
extensionally equal to $\tgte'$, given that both $\Binds$ and
$\Binds'$ are satisfiable. While we conjecture that this result is
true, a proof of this property out of our reach, at present. There are
at least two difficulties. First, a coherence result of this form is
unknown for qualified type systems in a call-by-value setting. In an
unpublished paper, Jones~\cite{jones93coherencefor} proves a coherence
result for OML, but his techique only applies to call-by-name
programs. Jones also does not consider reasoning about coherence based
on an equational theory for the evidence functions (these functions
correspond to our binds).
So, proving the ideal coherence theorem would require both
generalizing Jones' approach to call-by-value and then extending it
with support for equational reasoning about evidence. In the meantime,
Theorem~\ref{thm:coherence} provides good assurance and lays the
foundation for future work in this direction. \mwh{be careful how we
  refer to this result earlier}


\paragraph*{Defining and analyzing principality.} We introduce a notion of
 principal polymonads that corresponds to Tate's ``principalled
 productoids.'' Informally, in a principal polymonad, if there is more
 than one way to combine pairs of computations in the set $F$ (e.g.,
 $\bind{(M,M')}{M_1}$ and $\bind{(M,M')}{M_2}$), then there must be a
 ``best'' way to combine them. This best way is called the principal
 join of $F$, and all other ways to combine the functors are related
 to the principal join by morphisms. All the polymonadic libraries we
 have encountered so far are principal polymonads. It is worth
 emphasizing that principality does not correspond to functional
 dependency---it is perfectly reasonable to combine $\mconst$ and
 $\mconst'$ in multiple ways, and indeed, for applications like
 sub-effecting, this expressiveness is important. We only require that
 there be an ordering among the choices. In the definition below, we
 take $\downarrow\!\!\mathcal{M}$ to be set of ground instances of all
 constructors in $\mathcal{M}$.  

\begin{definition}[Principal polymonad]
  A polymonad $(\mathcal{M}, \Sigma)$ is a \emph{principal polymonad}
  if and only if for any set $F \subseteq \downarrow\!\!\mathcal{M}^2$, and any
  $\aset{\mconst_1, \mconst_2}\subseteq\downarrow\!\!\mathcal{M}$ such
  $\aset{\tbind{(\mconst, \mconst')}{\mconst_1} \mid (\mconst,\mconst') \in F} \subseteq
  \Sigma$ and $\aset{\tbind{(\mconst, \mconst')}{\mconst_2} \mid (\mconst,\mconst') \in F}
  \subseteq \Sigma$, then there exists $\hat\mconst \in \downarrow\!\!\mathcal{M}$ such
  that $\aset{\morph{\hat\mconst}{\mconst_1}, \morph{\hat\mconst}{\mconst_2}}
  \subseteq \Sigma$, and $\aset{\tbind{(\mconst,\mconst')}{\hat\mconst} \mid (\mconst,\mconst') \in
    F} \subseteq \Sigma$.  We call $\hat\mconst$ the principal join of $F$
  and write it as $\bigsqcup F$
\end{definition}


\newcommand\subscript[1]{{\mbox{\textit{\tiny #1}}}}

\begin{definition}[Graph-view of a constraint-bag $\Binds$]
  A graph-view $G_\Binds=(V,A,E_\rhd,E_{eq})$ of a constraint-bag
  $\Binds$ is a graph consisting of a set of vertices $V$, a vertex
  assignment $A: V -> m$, a set of directed edges $E_\rhd$, and a
  set of undirected edges $E_{eq}$, where:

  \begin{itemize}
  \item $V = \aset{\pi.0, \pi.1, \pi.2 \mid \pi \in \Binds}$, i.e.,
    each constraint contributes three vertices.

  \item $A(\pi.i) = m_i$ when $\pi = \tbind{(m_0,m_1)}{m_2}$, for all $\pi.i \in V$

  \item  $E_\rhd  =    \aset{(\pi.0,\pi.2), (\pi.1, \pi.2) \mid \pi \in \Binds}$

  \item $E_{eq}    = \aset{(v,v') \mid v,v' \in V ~\wedge~v\neq v'\wedge \exists \mvar.\mvar=A(v)=A(v')}$ 
  \end{itemize}
\end{definition}

\noindent\textbf{Notation} We use $v$ in this section to stand for a
graph vertex, rather than a value in a program. We also make use of a
pictorial notation for graph views, distinguishing the two flavors of
edges in a graph. Each constraint $\pi \in \Binds$ induces two edges
in $E_\rhd$. These edges are drawn with solid lines, with a triangle
for orientation. Unification constraints arise from correlated
variable occurrences in multiple constraints---we
\begin{wrapfigure}{r}{3.5cm}
\vspace{-5ex}
\begin{tiny}
  \[\nqquad
    \xymatrix@C=1em@R=0.5em{
      m_1 \ar@{-}[dr]   &             &                   &  m_2\ar@{-}[dr]  \\
      \mvar \ar@{-}[r]  & \rhd\ar@{-}[r]  & \mvar'\ar@{:}[r] & \mvar'\ar@{-}[r] & \rhd\ar@{-}[r] & \mvar\ar@/^/@{:}[lllll] 
    }
  \]
\end{tiny}
\vspace{-7ex}
\end{wrapfigure} 
depict these with double dotted lines. For
example, the pair of constraints $\tbind{(m_1, \mvar)}{\mvar'},
\tbind{(m_2,\mvar')}{\mvar}$ contributes four unification edges, two
for $\mvar$ and two for $\mvar'$. We show its graph view alongside.

Unification constraints reflect the dataflow in a program. Referring
back to Figure~\ref{fig:ssyntaxrules}, in a principal derivation using
(TS-App), correlated occurrences of unification variables for $m_4$ in
the constraints indicate how the two binds operators compose. The
following definition captures this dataflow and shows how to interpret
the composition of bind constraints using unification edges as a
lambda term (in the expected way).\footnote{Note, for the purposes of
  our coherence argument, unification constraints between value-type
  variables $a$ are irrelevant. Such variables may occur in two kinds
  of contexts. First, they may constrain some value type in the
  program, but these do not depend on the solutions to polymonadic
  constraints. Second, they may constrain some index of a polymonadic
  constructor; but, as mentioned previously, these indices are phantom
  and do not influence the semantics of elaborated terms.}

\begin{definition}[Functional view of a flow edge]
Given a constraint graph $G = (V, A, E_\rhd, E_{eq})$, an edge
$\eta=(\pi.2, \pi'.i) \in E_{eq}$, where $i \in \aset{0,1}$ and
$\pi\neq\pi'$ is called a \emph{flow edge}. The flow edge $\eta$ has a
functional interpretation $F_G(\eta)$ defined as follows:\\[-2ex]
\[\begin{array}{lcl}
$If$~~i=0,~~
  F_G(\eta) & = & \lambda (x@A(\pi.0)~a)~(y@a->A(\pi.1)~b)~(z@b->A(\pi'.1)~c).\\
          &   & ~~\sfont{bind}_{A(\pi'.0),A(\pi'.1),A(\pi'.2)}(\sfont{bind}_{A(\pi.0),A(\pi.1),A(\pi.2)}~x~y)~z\\

$If$~~i=1,~~
  F_G(\eta) & = & \lambda (x@A(\pi'.0)~a)~(y@a -> A(\pi.0)~b)~(z@b->A(\pi.1)~c).\\
          &   & ~~\sfont{bind}_{A(\pi'.0),A(\pi'.1),A(\pi'.2)}~x~(\lambda a.\sfont{bind}_{A(\pi.0),A(\pi.1),A(\pi.2)}~(y~a)~z)
\end{array}\]
\end{definition}

We can now define our ambiguity check---a graph is unambiguous if it
contains a sub-graph that has no cyclic dataflows, and where open
variables only occur as intermediate variables in a sequence of binds. 

\begin{definition}[Unambiguous constraints]
\label{def:unambiguous}
Given $G_\Binds=(V,A,E_\rhd,E_{eq})$, the predicate $\unamb{\Binds}{\Gamma}{t}$
holds if and only if there exists $E_{eq}' \subseteq
E_{eq}$, such that in the graph $G'=(V,A,E_\rhd,E_{eq}')$ all of the
following are true.

\begin{enumerate}
\item For all $\pi \in \Binds$, there is no path from $\pi.2$ to
  $\pi.0$ or $\pi.1$.

\item For all $v \in V$, if $A(v)\in \ftv{\Binds} \setminus
  \ftv{\Gamma,t}$, then there exists a flow edge that connects to
  $v$.
\end{enumerate}

\noindent We call $G'$ a \emph{core} of $G_\Binds$.
\end{definition}

\begin{definition}[Solution to a constraint graph]
For a polymonadic signature $(\mathcal{M}, \Sigma)$, a solution
to a constraint graph $G=(V, A, E_\rhd, E_{eq})$, is a vertex assignment
$S : V -> \mathcal{M}$ such that all of the following are true.

\begin{enumerate}
\item For all $v \in V$, if $A(v) \in \mathcal{M}$ then $S(v)=A(v)$

\item For all $(v_1,v_2) \in E_{eq}$, $S(v_1) = S(v_2)$.

\item For all $\aset{(\pi.0,\pi.2), (\pi.1,\pi.2)} \subseteq E_\rhd$, 
  $\tbind{(S(\pi.0),S(\pi.1))}{S(\pi.2)} \in \Sigma$.
\end{enumerate}

\noindent We say that two solutions $S_1$ and $S_2$ to $G$ \emph{agree
  on} $\mvar$ if for all vertices $v \in V$ such that $A(v) = \mvar$,
$S_1(v) = S_2(v)$.
\end{definition}

Now we define $\cong_R$, a notion of equivalence of two solutions
which captures the idea that the 
differences in the solutions are only to the internal open variables
while not impacting the overall function computed by the binds in a
constraint. It is easy to check that $\cong_R$ is an equivalence
relation.

\begin{definition}[Equivalence of solutions]
Given a polymonad $(\mathcal{M},\Sigma)$ and constraint
graph $G=(V,A,$ $E_\rhd,E_{eq})$, two solutions $S_1$ and $S_2$ to $G$
are equivalent with respect to a set of variables $R$ (denoted $S_1
\cong_R S_2$) if and only if $S_1$ and $S_2$ agree on all $\mvar
  \in R$ and for
each vertex $v \in V$ such that $S_1(v) \neq S_2(v)$ for all flow
edges $\eta$ incident on $v$, $F_{G_1}(\eta) = F_{G_2}(\eta)$,
where $G_i=(V, S_i, E_\rhd, E_{eq})$.
\end{definition}

\begin{theorem}[Coherence]
  \label{thm:coherence}
  For all principal polymonads, derivations $\Binds|\Gamma |- e : t
  \rew \tgte$ such that \\ $\unamb{\Binds}{\Gamma}{t}$, and for any two
  solutions $S$ and $S'$ to $G_\Binds$ that agree on
  $R=ftv(\Gamma,t)$, we have $S \cong_R S'$.
\end{theorem}

\noindent (Sketch; full version in appendix) The main idea is to show that all
solutions in the core of $G_\Binds$ are in the same equivalence class
(the solutions to the core include $S$ and $S'$). The proof proceeds
by induction on the number of vertices at which $S$ and $S'$
differ. For the main induction step, we take vertices in 
\begin{wrapfigure}{r}{5cm}
    \[\begin{tiny}\nquad
     \begin{array}{c}
       \underline{S/S'} \\
      \xymatrix@C=.1em@R=0.75em{
                                                                     & \mconst_1/\mconst_1'                           & \ldots & \mconst_2/\mconst_2' \\
                                                                     &   \tup\ar[u]                                  & \ldots\tup\ar[u]\ar@{-}[d]\ldots  &\tup\ar[u] \\
                                      \mconst_3/\mconst_3'\ar@{-}[ru] &  \mconst[A]/\mconst[B]\ar@{-}[u]\ar@{:}[r]\ar@{:}[dr]   & \ldots    & \ar@{:}[dl]\ar@{:}[l]\mconst[A]/\mconst[B]\ar@{-}[u] & \mconst_4/\mconst_4'\ar@{-}[ul] \\
                                                                     &   \eta_1\ar@{:}[u]                                        & \ldots\eta\ldots   & \eta_k\ar@{:}[u] \\
                                                                     & \mconst[A]/\mconst[B]\ar@{:}[u]\ar@{:}[r]\ar@{:}[ur]      & \ldots & \ar@{:}[ul]\ar@{:}[l]\mconst[A]/\mconst[B]\ar@{:}[u]        \\
                                                                     &   \tup\ar[u]                                    & \ldots\tup\ar[u]\ar@{-}[d]\ldots &  \tup\ar[u]                     \\
                                     \mconst_5\ar@{-}[ur]            &  \mconst_6\ar@{-}[u]                           & \ldots &   \mconst_7\ar@{-}[u]                        &  \mconst_8\ar@{-}[ul] \\
      }
    \end{array}
    \end{tiny}\]
\end{wrapfigure}
topological order, considering the least (in the order) set of vertices $Q$, all
related by unification constraints, and whose assignment in $S$ is
$\mconst[A]$ and in $S'$ is $\mconst[B]$, for some
$\mconst[A]\neq\mconst[B]$. The vertices in $Q$ are shown in the graph
alongside, all connected to each other by double dotted lines
(unification constraints), and their neighborhood is shown as
well. Since vertices are considered in topological order, all the
vertices below $Q$ in the graph have the same assignment in $S$ and in
$S'$. We build solutions $S_1$ and $S_1'$ from $S$ and $S'$
respectively, that instead assign the principal join $\mconst[J]=\bigsqcup \aset{(\mconst_5,\mconst_6),\ldots,(\mconst_7,\mconst_8)}$ to
the vertices in $Q$, where $S_1 \cong_R S_1'$ by the induction
hypothesis. Finally, we prove $S \cong_R S_1$ and $S' \cong_R S_1'$
by showing that the functional interpretation of each of the flow
edges $\eta_i$ are equal according to the polymonad laws, and conclude
$S \cong_R S'$ by transitivity.


\section{Simplification and solving} \label{sec:solve} 



\renewcommand{\dom}[1]{\textsf{dom}(#1)}
\newcommand{\transupperbounds}[2]{\mbox{\textit{trans-up-bnd}}_{#1}(#2)}

Before running a program, we must solve the constraints produced
during type inference, and apply the appropriate evidence for these
constraints in the elaborated program. We also perform
\emph{simplification} on constraints prior to generalization to make
types easier to read, but without compromising their utility. 

A simple syntactic transformation on constraints can make inferred
types easier to read. For example, we can hide duplicate constraints,
identity morphisms (which are trivially satisfiable), and constraints
that are entailed by the signature.
More substantially, we can find instantiations for open variables in a
constraint set before generalizing a type (and at the top-level,
before running a program). To do this, we introduce below a modified
version of (TS-Let) (from Figure~\ref{fig:ssyntaxrules}); a similar
modification is possible for (TS-Rec).  

\begin{small}
\[
   \inference{\prefix{\Binds_1}\Gamma v : \ty  \rew \tgte_1 &
              \bar\mvar,\bar{a} = \ftv{\Binds_1 => \ty} \setminus \ftv{\Gamma} \\
               \Binds_1 \simp[\bar\mvar \setminus \ftv{\ty}] \theta
               &
               (\sigma,\tgte_2) = \Gen{\Gamma}{\theta\Binds_1 =>
                 \ty, \tgte_1} &
              \prefix{\Binds}{\Gamma,x@\sigma} e : \tapp{m}\ty' \rew \tgte_3}
             {\prefix{\Binds}\Gamma \slet{x}{v}{e}  : \tapp{m}\ty' \rew \slet{x}{\tgte_2}{\tgte_3}}
\]
\end{small}

This rule employs the judgment $\Binds \simp \theta$, defined in
Figure~\ref{fig:decl-solving}, to simplify constraints by eliminating
some open variables in $\Binds$ (via the substitution $\theta$) before
type generalization. There are three main rules in the judgment,
(S-$\Uparrow$), (S-$\Downarrow$) and (S-$\sqcup$), while the last two
simply take the transitive closure. 

\begin{figure}[t!]
\[\small
\begin{array}{c}
\inference[S-$\Uparrow$]
          {\pi = \bind{(\mconst[Id],m)}{\mvar} \;\vee\; \pi =
            \bind{(m,\mconst[Id])}{\mvar} \\ \mvar \in \bar\mvar &
    \flowsFrom{\mvar}{\Binds,\Binds'} \neq \{\} \\
  \flowsTo{\mvar}{\Binds,\Binds'} =\{\} }
          {\Binds,\pi,\Binds' \simp \mvar \mapsto m}

\qquad

\inference[S-$\Downarrow$]
          {\pi = \bind{(\mconst[Id],\mvar)}{m} \;\vee\; \pi =
            \bind{(\mvar,\mconst[Id])}{m} \\ \mvar \in \bar\mvar &
    \flowsFrom{\mvar}{\Binds,\Binds'} = \{\} \\  \flowsTo{\mvar}{\Binds,\Binds'} \neq \{\} }
          {\Binds,\pi,\Binds' \simp \mvar \mapsto m}

\\\\
\inference[S-$\sqcup$]
          { F = \flowsTo{\mvar}{\Binds} \\ m \in F \Rightarrow m =
            \gm\\ \text{for some $\gm$} }
          {\Binds \simp \mvar \mapsto \bigsqcup F}

\qquad

\inference{\Binds \simp \theta \\ \theta\Binds \simp \theta'}
          {\Binds \simp \theta'\theta}

\qquad

\inference{} {\Binds \simp \cdot}
\end{array}\]
\[\small
\text{where~~}
\begin{array}{lcl} 
  \flowsTo{\mvar}{\Binds}   & = &  \aset{\,(m_1,m_2) \mid 
    \bind{(m_1,m_2)}{\mvar} \in \Binds \,} \\
  \flowsFrom{\mvar}{\Binds} & = &  \aset{\,m \mid \exists m'.\;~ \pi \in
    \Binds~ \wedge~ (\pi = \bind{(\mvar,m')}{m}~ \vee~ \pi = \bind{(m',\mvar)}{m})\,}  \\
\end{array}
\]
\caption{Eliminating open variables in constraints}
\label{fig:decl-solving}
\end{figure}

Rule (S-$\Uparrow$) solves monad variable $\mvar$ with monad $m$ if we
have a constraint $\pi =\bind{(\mconst[Id], m)}{\mvar}$, where the
only edges directed inwards to $\mvar$ are from $\mconst[Id]$ and $m$,
although there may be many out-edges from $\mvar$. (The case where
$\pi=\bind{(m,\mconst[Id])}{\mvar}$ is symmetric.) Such a constraint
can always be solved without loss of generality using an identity
morphism, which, by the polymonad laws is guaranteed to
exist. Moreover, by the closure law, any solution that chooses
$\mvar=m'$, for some $m'\neq m$ could just as well have chosen
$\mvar=m$. Thus, this rule does not impact solvability of the
costraints. Rule S-$\Downarrow$ follows similar reasoning in the
reverse direction.
Finally, we the rule (S-$\sqcup$) exploits the properties
of a principal polymonad. Here we have a variable $\mvar$ such that
all its in-edges are from pairs of ground constructors $\gm_i$, so we
can simply apply the join function to compute a solution for
$\mvar$. For a principal polymonad, if such a
solution exists, this simplification does not impact solvability of
the rest of the constraint graph.

\paragraph*{Example.} Recall the information flow example we gave in
Section~\ref{sec:ist-example}, in Figure~\ref{fig:ist}. Its principal type is the following,
which is hardly readable:
\[\small\begin{array}{l@{~}l}
\multicolumn{2}{l}{\forall \bar\mvar_i, a_1, a_2. \Binds_0 =>
 \lfont{\intref\;a_1 \rightarrow \intref\;a_2 \rightarrow
   \mvar_{27}\;()}} \\
\text{where } \Binds_0 = &
\bind{(\mId,\mvar_{3})}{\mvar_{2}}, \bind{(\mId,\mIST\; \Hi\; a_2)}{\mvar_{3}}, \bind{(\mvar_{26},\mId)}{\mvar_{4}},
\bind{(\mId,\mId)}{\mvar_{4}}, \\
& \bind{(\mvar_{8},\mvar_{4})}{\mvar_{6}}, \bind{(\mId,\mvar_{9})}{\mvar_{8}}, \bind{(\mId,\mId)}{\mvar_{9}},
\bind{(\mvar_{11},\mvar_{25})}{\mvar_{26}}, \\
& \bind{(\mId,\mvar_{12})}{\mvar_{11}}, \bind{(\mId,\mIST\; \Hi\;
  a_1)}{\mvar_{12}}, \bind{(\mvar_{17},\mvar_{23})}{\mvar_{25}},
\bind{(\mvar_{14},\mvar_{18})}{\mvar_{17}}, \\
& \bind{(\mId,\mId)}{\mvar_{18}}, \bind{(\mId,\mvar_{15})}{\mvar_{14}},
\bind{(\mId,\mId)}{\mvar_{15}},
\bind{(\mvar_{20},\mvar_{24})}{\mvar_{23}}, \\
& \bind{(\mId,\mIST\; a_1\; \Lo)}{\mvar_{24}}, \bind{(\mId,\mvar_{21})}{\mvar_{20}},
\bind{(\mId,\mId)}{\mvar_{21}}.
\end{array}\]


After applying (S-$\Uparrow$) and (S-$\Downarrow$) several times, and
then hiding redundant constraints, we simplify $\Binds_0$ to $\Binds$
which contains only three constraints. If we had fixed $a_1$ and $a_2$
(the labels of the function parameters) to $\Hi$ and $\Lo$,
respectively, we could do even better. The three constraints would be 
$\bind{(\mIST\,\Hi\,\Lo,\mvar_{6})}{\mvar_{27}},
\bind{(\mId,\mId)}{\mvar_6},\bind{(\mIST\,\Hi\,\Hi,\mIST\,\Hi\,\Lo)}{\mvar_{6}}$.
Then, applying (S-$\sqcup$) to $\mvar_6$ we would get $\mvar_{6}
\mapsto \mIST\,\Hi\,\Hi$, which when applied to the other constraints
leaves only $\bind{(\mIST\,\Hi\,\Lo,\mIST\,\Hi\,\Hi)}{\mvar_{27}}$,
which cannot be simplified further, since $\mvar_{27}$ appears in the
result type.

Pleasingly, this process yields a simpler type that can be used in the
same contexts as the original principal type, so we are not
compromising the generality of the code by simplifying its type.

\begin{lemma}[Simplification improves types] 
\label{lem:simplification} 
For a principal polymonad, given $\sigma$ and $\sigma'$ where $\sigma$ is $\forall
\bar{a}\bar{\mvar}. \Binds => \ty$ and $\sigma'$ is an \emph{improvement} of
$\sigma$, having form $\forall \bar{a'}\bar{\mvar'}. \theta\Binds => \ty$ where $\Binds
\simp \theta$ and $\bar{a'}\bar{\mvar'} = (\bar{a}\bar{\mvar}) - dom(\theta)$.  Then
for all $\Binds'', \Gamma, x, e, m, \tau$, if
$\prefix{\Binds''}{\Gamma,x@\sigma} e : m\,\tau$ such that $|=
\Binds''$ then there exists some $\Binds'''$ such that
$\prefix{\Binds'''}{\Gamma,x@\sigma'} e : m\,\tau$ and $|= \Binds'''$.
\end{lemma}

Note that our $\simp$ relation is non-deterministic in the way it
picks constraints to analyze, and also in the order in which rules are
applied. In practice, for an acyclic constraint graph, one could
consider nodes in the graph in topological order and, say,
apply (S-$\sqcup$) first, since, if it succeeds, it eliminates a
variable. For principal polymonads and acyclic constraint
graphs, this process would always terminate.

However, if unification constraints induce cycles in the
constraint graph, simply computing joins as solutions to internal
variables may not work. This should
not come as a surprise. In general, finding solutions to arbitrary
polymonadic constraints is undecidable, since, in the limit, they can
be used to encode the correctness of programs with general
recursion. Nevertheless, simple heuristics such as unrolling cycles in
the constraint graph a few times may provide good mileage, as would
the use of domain-specific solvers for particular polymonads, and such
approaches are justified by our coherence proof. 

\section{Related work and conclusions}

This paper has presented \emph{polymonads}, a generalization of monads
and morphisms, which, by virtue of their relationship to Tate's
\emph{productoids}, are extremely powerful, subsuming monads,
parameterized monads, and several other interesting
constructions. Thanks to supporting algorithms for (principal) type
inference, (provably coherent) elaboration, and
(generality-preserving) simplification (none of which Tate considers),
this power comes with strong supports for the programmer. Like monads
before them, we believe polymonads can become a useful and important
element in the functional programmer's toolkit.

Constructions resembling polymonads have already begun to creep into
languages like Haskell. Notably, Kmett's
\texttt{Control.Monad.Parameterized} Haskell package~\cite{kmett}
provides a type class for bind-like operators that have a signature
resembling our $\tbind{(m_1,m_2)}{m_3}$.  One key limitation is that
Kmett's binds must be \emph{functionally dependent}: $m_3$ must be
functionally determined from $m_1$ and $m_2$.  As such, it is not possible to program
morphisms between different  constructors, i.e., the pair of
binds $\tbind{(m_1,\mId)}{m_2}$ and $\tbind{(m_1,\mId)}{m_3}$ would be
forbidden, so there would be no way to convert from $m_1$ to $m_2$ and
from $m_1$ to $m_3$ in the same program. Kmett also requires units
into $\mId$, which may later be 
lifted, but such lifting only works for first-order code before running afoul
of Haskell's ambiguity restriction. Polymonads do not have either
limitation.  Kmett does not discuss laws that should govern the proper
use of non-uniform binds. As such, our work provides the formal basis
to design and reason about libraries that functional
programmers have already begun developing.

While polymonads subsume a wide range of prior monad-like
constructions, and indeed can express any system of \emph{producer
  effects}~\cite{tate12productors}, as might be expected, other researchers have explored
generalizing monadic effects along other dimensions that are
incomparable to polymonads. For example, Altenkirch et
al.~\cite{Altenkirch10relative} consider \emph{relative monads} that
are not endofunctors; each polymonad constructor must be an
endofunctor. Uustalu and Vene~\cite{Uustalu08comonad} suggest
structuring computations comonadically, particularly to work with
context-dependent computations. This suggests a loose connection with
our encoding of contextual effects as a polymonad, and raises the
possibility of a ``co-polymonad'', something we leave for the
future. Still other generalizations include reasoning about effects
equationally using Lawvere theories~\cite{plotkin01semantic} or with
arrows~\cite{Hughes00arrows}---while each of these generalize monadic
constructions, they appear incomparable in expressiveness to
polymonads. A common framework to unify all these treatments of
effects remains an active area of research---polymonads are a useful
addition to the discourse, covering at least one large area of the
vast design space.

\let\oldthebibliography=\thebibliography
\let\endoldthebibliography=\endthebibliography
\renewenvironment{thebibliography}[1]{%
  \begin{oldthebibliography}{#1}%
    \vskip -20mm
    \setlength{\parskip}{0ex}%
    \setlength{\itemsep}{.25ex}%
  }%
  {
  \end{oldthebibliography}%
}

\begin{small}
\bibliographystyle{eptcs}
\bibliography{monadic}
\end{small}

\pagebreak
\appendix
\begin{center}
\textbf{Appendix}
\end{center}

\section{Polymonads are productoids and vice versa}
\label{sec:productoids}

%
%
Given a polymonad $(\mathcal{M},\Sigma)$, we can construct a 4-tuple
$(\mconstrs, U, L, B)$ as follows:

\begin{description}
\item[(Units)] $U = \aset{(\lambda x. \mybind~x~(\lambda y.y))\colon \kw{a -> M a} \mid \mybind\colon\bind{(\mId,\mId)}{M} \in \Sigma}$, 

\item[(Lifts)] $L = \aset{(\lambda x. \mybind~x~(\lambda y.y))\colon \kw{M a -> N a} \mid \mybind\colon\morph{\mconst}{\mconst[N]} \in \Sigma}$, 

\item[(Binds)] The set $B = \Sigma- \aset{\mybind \mid \mybind\colon\bind{(\mId,\mId)}{M}~\mbox{or}~\mybind\colon\bind{(M,\mId)}{N} \in \Sigma}$.
\end{description}
It is fairly easy to show that the above structure satisfies
generalizations of the familiar laws for monads and monad
morphisms.

\begin{theorem}
\label{thm:bind-as-mubl}
Given a polymonad  $(\mathcal{M},\Sigma)$, the induced 4-tuple
$(\mconstrs, U, L, B)$ satisfies the following properties.

\begin{description}
\item [(Left unit)] $\forall\myunit\in U, \mybind\in B$. if $\kw{unit: forall a. a -> M a}$ and 
$\kw{bind:} \bind{(M,N)}{N}$ then 
$
\mybind~(\myunit~e)~ f = f(e)
$
where $e\colon\tau$ and $f\colon\tau~\kw{-> N}~\tau'$.

\item [(Right unit)] $\forall\myunit\in U, \mybind\in B$. if $\kw{unit: forall a. a -> N a}$ and 
$\kw{bind:} \bind{(M,N)}{M}$ then
$\mybind~m~(\myunit) = m$ 
where $m\colon \kw{M}~\tau$.

\item [(Associativity)] $\forall\mybind_1, \mybind_2, \mybind_3,
  \mybind_4 \in B$. if 
$\mybind_1: \bind{(M,N)}{P}$, \newline
$\mybind_2: \bind{(P,R)}{T}$, 
$\mybind_3: \bind{(M,S)}{T}$, and 
$\mybind_4: \bind{(N,R)}{S}$ then \newline
$\mybind_2~(\mybind_1~m~f)~g
=
\mybind_3~m~(\lambda x. \mybind_4~(f~x)~g)$ \newline
where $m\colon \kw{M}~\tau$, $f\colon \tau~\kw{-> N}~\tau'$
and $g\colon\tau'~\kw{-> R}~\tau''$ 

\item [(Morphism 1)]
$\forall\myunit_{1}, \myunit_{2}\in U, \mylift\in L$. if $\myunit_{1}\kw{: forall a. a -> M a}$, 
 $\myunit_2\kw{: forall a. a -> N a}$ and $\kw{lift: forall a. M a -> N a}$ then 
$\mylift~(\myunit_1~e) 
=
\munit_2~e$ 
where $e\colon\tau$.

\item [(Morphism 2)]
$\forall\mybind_1, \mybind_2\in B, \mylift_1, \mylift_2, \mylift_3\in L$. if
$\mybind_1: \bind{(M,P)}{S}$, \newline
$\mybind_2: \bind{(N,Q)}{T}$, 
$\mylift_1\kw{: forall a. M a -> N a}$,
$\mylift_2\kw{: forall a. P a -> Q a}$ and 
$\mylift_3\kw{: forall a. S a -> T a}$
then
$\mylift_3~(\mybind_1~m~f) 
=
\mybind_2~(\mylift_1~m)~(\lambda x.\mylift_2~(f~x))$ \newline
where $m\colon \kw{M}~\tau$ and $f\colon\tau~\kw{-> P}~\tau'$.
\end{description}
\end{theorem}


\newcommand\eff{E}
\newcommand\uniteff[1]{U}
\newcommand\joineff[3]{\ensuremath{(#1\mathbb{;}#2) \mapsto #3}}
\newcommand\effectoid{\ensuremath{(\eff, \uniteff{\_}, \leq,
    \joineff{\_}{\_}{\_})}}

Now we show how this definition can be used to relate polymonads to
Tate's \emph{productoids}~\cite{tate12productors}.
\iffull
\else
The definition of a productoid is driven by an underlying algebraic structure: the effectoid~\cite[Theorem~1]{tate12productors}. 
An effectoid $(E, U, \leq, \mapsto)$ is a set $E$,
with an identified subset $U\subseteq E$ and relations $\mathord{\leq}\subseteq
E\times E$ and $(\_;\_)\mapsto \_\subseteq E\times E\times E$, that
satisfies six monoid-like conditions.  
It is possible to show that a polymonad directly induces an effectoid structure and hence a productoid.
\fi

\begin{lemma}
Given a polymonad $(\mconstrs, U, L, B)$ we can define an effectoid \effectoid{} as follows.
\[
\begin{array}{ll}
\eff = \mconstrs
&
\uniteff{\_} \; = \{ \kw{M} \mid \kw{unit: a -> M a} \in U\}\\
\leq\ = \{ (\kw{M},\kw{N}) \mid \kw{lift: M a -> N a} \in L\}~~~~
&
\joineff{\_}{\_}{\_}\ = \{ (\kw{M}, \kw{N}, \kw{P}) \mid \bind{(M,N)}{P} \in B\}
\end{array}
\]
\end{lemma}

\iffull
\else
\begin{lemma}
Every polymonad gives rise to a productoid.
\end{lemma}
\begin{proof}
We have shown that a polymonad gives rise to an effectoid. Given an
effectoid \effectoid\  a productoid is defined as a collection of functors
indexed by the collection $E$, and three collections of natural
transformations indexed by the three relations. These functors and
natural transformations are required to satisfy five addition
properties~\cite[Theorem 2]{tate12productors}. The five properties are
the five properties of Theorem~\ref{thm:bind-as-mubl}, so the proof is
immediate.
\end{proof}
\fi
Interestingly, we can identify conditions where the opposite direction also holds.

\begin{lemma}
A productoid $(\mathbf{C}, \{F_e\colon\mathbf{C}\to\mathbf{C}\}_{e\in
  E}, \{\eta\colon 1\Rightarrow F_e\}_{e\in U}, 
 \{\mu\colon F_{e_1}\circ F_{e_2}\Rightarrow
 F_{e_{3}}\}_{(e_1;e_2)\mapsto e_3}, 
 \{\sigma\colon F_{e_1}\Rightarrow F_{e_{2}}\}_{e_{1}\leq e_{2}})$
that in addition satisfies the following conditions
gives rise to a polymonad. 

\noindent\begin{tabular}{@{}ll}
1. $\mId \in E$ and $F_\mId = 1$ & 4. For all $e\in E, \mu\colon F_e\circ 1\Rightarrow F_e=\mId$\\
2. $\mId\in U$ & 5. $\joineff{e_1}{e_2}{e} \;\wedge\; e_1' \leq e_1 \;=>\; \joineff{e_1'}{e_2}{e}$ \\
3. For all $e\in E$, $\joineff{e}{\mId}{e}$ \quad\quad\quad&
6. $\joineff{e_1}{e_2}{e} \;\wedge\; e_2' \leq e_2 \;=>\; \joineff{e_1}{e_2'}{e}$
\end{tabular}
\end{lemma}
These additional conditions are fairly mild: (1)-(4) simply ensure
that the $\mId$ element is interpreted as the identity
functor. Conditions (5)-(6) are also quite straightforward; certainly
if the category is cartesian closed then the extra natural
transformations are always defined.

\section{Coherence of solutions}

\begin{lemma}[Solutions to a core]
\label{lem:acycliccore}
For a polymonad $(\mathcal{M},\Sigma)$, and a constraint graph $G$
with a core $G'$, the set of all solutions $\mathcal{S}'$ to
$G'$ includes all the solutions $\mathcal{S}$ of $G$.
\begin{proof}
(Sketch) This is easy to see, since $G'$ differs from $G$ only in that it
includes fewer unification constraints. So, all solutions to $G$ are
also solutions to $G'$.
\end{proof}
\end{lemma}

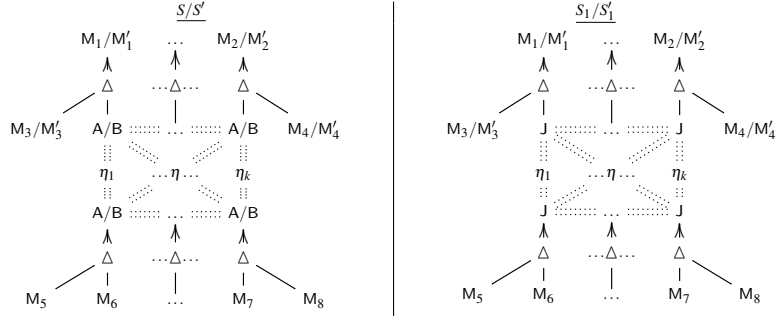
\begin{figure*}[t!]
    \[
    \begin{tiny}\nquad
     \begin{array}{c|c}
       \underline{S/S'} & \underline{S_1/S_1'} \\

      \xymatrix@C=.1em@R=0.75em{
                                                                     & \mconst_1/\mconst_1'                           & \ldots & \mconst_2/\mconst_2' \\
                                                                     &   \tup\ar[u]                                  & \ldots\tup\ar[u]\ar@{-}[d]\ldots  &\tup\ar[u] \\
                                      \mconst_3/\mconst_3'\ar@{-}[ru] &  \mconst[A]/\mconst[B]\ar@{-}[u]\ar@{:}[r]\ar@{:}[dr]   & \ldots    & \ar@{:}[dl]\ar@{:}[l]\mconst[A]/\mconst[B]\ar@{-}[u] & \mconst_4/\mconst_4'\ar@{-}[ul] \\
                                                                     &   \eta_1\ar@{:}[u]                                        & \ldots\eta\ldots   & \eta_k\ar@{:}[u] \\
                                                                     & \mconst[A]/\mconst[B]\ar@{:}[u]\ar@{:}[r]\ar@{:}[ur]      & \ldots & \ar@{:}[ul]\ar@{:}[l]\mconst[A]/\mconst[B]\ar@{:}[u]        \\
                                                                     &   \tup\ar[u]                                    & \ldots\tup\ar[u]\ar@{-}[d]\ldots &  \tup\ar[u]                     \\
                                     \mconst_5\ar@{-}[ur]            &  \mconst_6\ar@{-}[u]                           & \ldots &   \mconst_7\ar@{-}[u]                        &  \mconst_8\ar@{-}[ul] \\
      }

      \qquad&\qquad

      \xymatrix@C=.1em@R=0.75em{
                                                                     & \mconst_1/\mconst_1'                           & \ldots & \mconst_2/\mconst_2' \\
                                                                     &   \tup\ar[u]                                  & \ldots\tup\ar[u]\ar@{-}[d]\ldots  &\tup\ar[u] \\
                                      \mconst_3/\mconst_3'\ar@{-}[ru] &  \mconst[J]\ar@{-}[u]\ar@{:}[r]\ar@{:}[dr]   & \ldots    & \ar@{:}[dl]\ar@{:}[l]\mconst[J]\ar@{-}[u] & \mconst_4/\mconst_4'\ar@{-}[ul] \\
                                                                     &   \eta_1\ar@{:}[u]                                        & \ldots\eta\ldots   & \eta_k\ar@{:}[u] \\
                                                                     & \mconst[J]\ar@{:}[u]\ar@{:}[r]\ar@{:}[ur]      & \ldots & \ar@{:}[ul]\ar@{:}[l]\mconst[J]\ar@{:}[u]        \\
                                                                     &   \tup\ar[u]                                    & \ldots\tup\ar[u]\ar@{-}[d]\ldots &  \tup\ar[u]                     \\
                                     \mconst_5\ar@{-}[ur]            &  \mconst_6\ar@{-}[u]                           & \ldots &   \mconst_7\ar@{-}[u]                        &  \mconst_8\ar@{-}[ul] \\
      }

     \end{array}
    \end{tiny} 
\]
\caption{Constraint graphs used to illustrate the proof of coherence
  (Theorem~\ref{thm:coherence})}
\label{fig:graphs}
\end{figure*}

\begin{theorem}[Coherence]
  For all principal polymonads, derivations $\Binds|\Gamma |- e : t
  \rew \tgte$ such that \\ $\unamb{\Binds}{\Gamma}{t}$, and for any two
  solutions $S$ and $S'$ to $G_\Binds$ that agree on
  $R=ftv(\Gamma,t)$, we have $S \cong_R S'$.
\begin{proof}
    We consider the set $\mathcal{S}$ of all solutions to the 
    core of $G_\Binds$ that agree on $ftv(\Gamma,t)$, and prove that
    all these solutions are in the same equivalence class. By
    Lemma~\ref{lem:acycliccore}, $\aset{S,S'} \subseteq
    \mathcal{S}$, establishing our goal. 

    Let $G=(V,A,E_\rhd,E_{eq})$ be a core of $G_\Binds$ and let $S$ and $S'$
    be arbitrary elements of $\mathcal{S}$. $S$ and $S'$ may only
    differ on the open variables of $\Binds$. Since $G$ is
    unambiguous, the nodes associated with these variables all have
    non-zero in- and out-degree. Let $U_{S,S'} = \aset{v \mid v \in V
      \;\wedge\; S(v) \neq S'(v)}$; the proof proceeds by induction
    on the the size of $U$.

    \noindent\textbf{Base case $|U_{S,S'}|=0$:} Trivial, since we have
    $S(v) = S'(v)$, for all $v$.

    \noindent\textbf{Induction step $|U_{S,S'}=i|$:} From the induction hypothesis:
    All solutions $S_1$ and $S_1'$ such that $|U_{S_1,S_1'}| < i$, we
    have $S_1 \cong S_1'$.  

    Topologically sort $G$, such that all vertices in the same
    connected component following edges in $E_{eq}$ have the same
    index, and each vertex $v$ is assigned an index greater than the
    index of all vertices $v'$ such that $(v,v')$ is an edge in
    $E_\rhd$. That is, ``leaf'' nodes have the highest indices. 

    Pick a vertex $v$ with the maximal index, such that $S(v) =
    \mconst[A]$ and $S'(v) = \mconst[B]$, for $\mconst[A] \neq
    \mconst[B]$, and let $I$ be the set of vertices reachable from $v$
    via unification edges. Since both $S$ and $S'$ are solutions,
    there must exist an open variable $\mvar$ such that $A(v) =\mvar$,
    and since $G$ is a core, there must be some non-empty set of flow
    edges incident on $v$.

    Thus, the neighborhood of $v$ in the graphs $G$, under assignment
    $S$ and $S'$ has a shape as shown in graph at left in
    Figure~\ref{fig:graphs}. All the nodes in $I$ are shown connected
    by double dotted lines---they each have assignment
    $\mconst[A]/\mconst[B]$ in $S/S'$.
    Since all the nodes in $I$ have an index greater than
    the index of any variable that differs among $S$ and $S'$, all
    their immediate predecessors have identical assignments in the two
    solutions (i.e,. $\mconst_5, \ldots, \mconst_8$). However, the
    other assignments may differ, (e.g., the top-left node could be
    assigned $\mconst_1$ in $S$ and $\mconst_1'$ in $S'$, etc.)
    Each flow-edge $\aset{\eta_1, \ldots, \eta_k}$ incident upon one
    of the nodes with the same index as $v$ is also labeled.

    Now, since we have a principal polymonad, there exists a principal
    join of $\aset{(\mconst_5, \mconst_6), \ldots, (\mconst_7,
      \mconst_8)}$---call it $\mconst[J]$.
    Consider the assigment $S_1$ (resp. $S_1'$) that differs from
    $S$ (resp. $S'$) only by assigning $\mconst[J]$ to each vertex
    in $I$ instead of $\mconst[A]$ (resp. $\mconst[B]$).

    We first show that $S_1$ (resp. $S_1'$) is a solution and that $S
    \cong S_1$ (resp. $S' \cong S_1'$). Then, we note that since $S_1$
    and $S_1'$ agree on all the vertices in $I$, $|U_{S_1,S_1'}| < i$,
    so we apply the induction hypothesis to show that $S_1 \cong S_1'$
    and conclude with transitivity of $\cong$.

    To show that $S_1$ (resp $S_1'$) is a solution, since $\mconst[J]$
    is a join of $\mconst_5,\mconst_6, \ldots$, then
    $\aset{\bind{(\mconst_5, \mconst_6)}{\mconst[J]}, \ldots,
      \bind{(\mconst_7, \mconst_8)}{\mconst[J]}}$ all exist, as well
    as $\morph{\mconst[J]}{\mconst[A]}$
    (resp. $\morph{\mconst[J]}{\mconst[B]}$). By the Closure property,
    for every $\bind{(\mconst[M],\mconst[A])}{\mconst[M']}$ (resp.
    $\mconst[B]$) there also exists
    $\bind{(\mconst[M],\mconst[J])}{\mconst[M']}$. Thus, the
    assignment of $\mconst[J]$ to $I$ is valid for a solution.

    To show that $S \cong S_1$ (resp. $S' \cong S_1'$), we have to
    show that $F_S(\eta_i) = F_{S_1}(\eta_i)$ (resp. $F_{S'}(\eta_i) =
    F_{S_1'}(\eta_i)$), for all $i$. Taking $\eta_k$ as a
    representative case (the other cases are similar), we need to show
    the identity below, which is an immediate corollary of
    Associativity 1 and 2 (resp. for $\mconst[B],
    \mconst_4',\mconst_2'$).

    \[\begin{array}{l}
      \sfont{bind}_{\mconst[A],\mconst_4,\mconst_2}(\sfont{bind}_{\mconst_7,\mconst_8,\mconst[A]}~x~y)~z =
      \sfont{bind}_{\mconst[J],\mconst_4,\mconst_2}(\sfont{bind}_{\mconst_7,\mconst_8,\mconst[J]}~x~y)~z
      \end{array}\]
  \end{proof}
\end{theorem}

\end{document}

\begin{figure*}[th!]
{\begin{small}\[
\begin{array}{c}
  \fbox{$\Binds |= \Binds'$} \qquad
  \inference{\forall \pi \in \Binds'. \pi \in \Binds \vee \pi \in \Sigma}
              {\Binds |= \Binds'}
\\\\
  \fbox{$\Binds |= \sigma \tygt \ty$} \qquad
  \inference{
    \theta = [\bar \tau/\bar{a}][\bar{m}/\bar{\mvar}] & \Binds |= \theta\Binds_1}
            {\Binds |= (\tscheme{\bar{a}\bar{\mvar}}{\Binds_1}{\ty}) \,\tygt\,
              {\theta\ty}}[(TS-Inst)]
\\\\
  \fbox{$\prefix\Binds\Gamma v : \ty$} 
\\\\  
  \inference{v\in\aset{x,c} & \Binds |= \Gamma(v) \tygt \ty}
            {\prefix\Binds\Gamma v : \ty}[(TS-XC)]
  \quad
  \inference{\prefix\Binds{\Gamma,x@\ty_1} e : \tapp{m}{\ty_2}}
            {\prefix\Binds\Gamma \slam{x}{e} : \tfun{\ty_1}{\tapp{m}{\ty_2}}}[(TS-Lam)]
\\\\
  \fbox{$\prefix\Binds\Gamma e : \tapp{m}\ty$} 
\\\\
  \inference{\prefix\Binds\Gamma v : \ty}
            {\prefix{\Binds,\morph{\mconst[Id]}{m}}\Gamma v : \tapp{m}{\ty}}[(TS-V)]
\\\\
  \inference{\prefix{\Binds_1}{\Gamma,f@\ty} v : \ty &
    \prefix\Binds{\Gamma,f@\Gen{\Gamma}{\Binds_1 => \ty}} e : \tapp{m}{\ty'}}
            {\prefix\Binds\Gamma \sletrec{f}{v}{e} : \tapp{m}{\ty'}}[(TS-Rec)]
\\\\
  \inference{\prefix{\Binds_1}\Gamma v : \ty &
    \prefix\Binds{\Gamma,x@\Gen{\Gamma}{\Binds_1 => \ty}} e : \tapp{m}\ty'}
            {\prefix\Binds\Gamma \slet{x}{v}{e}  : \tapp{m}\ty'}[(TS-Let)]
\\\\
  \inference{ \prefix\Binds\Gamma e_1 : \tapp{m_1}{\ty_1} &
    \prefix\Binds{\Gamma,x@\ty_1} e_2 : \tapp{m_2}{\ty_2} &
    e_1 \neq v &\Binds |= \bind{(m_1,m_2)}{m_3}}
            {\prefix\Binds\Gamma \slet{x}{e_1}{e_2} : \tapp{m_3}{\ty_2}}[(TS-Do)]
\\\\
  \inference{\prefix\Binds\Gamma  e_1 : \tapp{m_1}{(\tfun{\ty_2}{\tapp{m_3}{\ty}})} &
    \prefix\Binds\Gamma  e_2 : \tapp{m_2}{\ty_2}  \\
    \Binds |= \bind{(m_1,m_4)}{m_5} &
    \Binds |= \bind{(m_2,m_3)}{m_4}  }
            {\prefix\Binds\Gamma  \eapp{e_1}{e_2} : \tapp{m_5}{\ty} }[(TS-App)]
\\\\
  \inference{\prefix\Binds\Gamma e_1 : \tapp{m_1}\kw{bool} &
    \prefix\Binds\Gamma e_2 : \tapp{m_2}{\ty} &
    \prefix\Binds\Gamma e_3 : \tapp{m_3}{\ty} \\
    \Binds |= \morph{m_2}{m},  \morph{m_3}{m}, \bind{(m_1,m)}{m'}}
                                            {\prefix\Binds\Gamma \sif{e_1}{e_2}{e_3} : \tapp{m'}{\ty}}[(TS-If)]
  \end{array}
\]\end{small}}
  \caption{Syntax-directed type rules for \lang{}, where the signature $\Sigma$
    is an implicit parameter. Generalization is defined as 
    $\Gen{\Gamma}{\Binds => \ty} = \forall (\ftv{\Binds => \ty} \setminus \ftv{\Gamma}). \Binds => \ty$}
  \label{fig:old-ssyntaxrules}
\end{figure*}

\subsection{Translation to Haskell}

Since we now have a mechanical translation to OML, we can actually
translate directly to Haskell, which uses OML's type inference
algorithm. We can define our primitives directly in Haskell as:
\begin{lstlisting}
newtype Id a = MkId a
class Bind m1 m2 m3 where 
  (>>>=) :: m1 a -> (a -> m2 b) -> m3 b

lift :: Bind m1 Id m2 => m1 a -> m2 a
lift e = e >>>= $ $ MkId

ret :: Bind Id Id m => a -> m a
ret x = lift (MkId x) 

app :: Bind m1 m4 m, Bind m2 m3 m4 => m1 (a -> m3 b) -> m2 a -> m b
app e1 e2 = e1 >>>=$ $ (\f ->$ $ e2 >>>=$ $ f)

cond :: Bind m2 Id m, Bind m3 Id m, Bind m1 m m' => 
        m1 bool -> (() -> m2 a) -> (() -> m3 a) -> m' a
cond e e1 e2 = e >>>=$ $ (\b -> if b then lift (e1 ()) else lift (e2 ()))
\end{lstlisting}
Of course, Haskell is lazy and the \lstinline|let| bindings in
Figure~\ref{fig:xlate-oml} need to be translated to strict bindings.
Also, we need to give a translation of all our monadic types and
bind definitions in $\Sigma$ to Haskell. We cannot define this precisely
since in
our system we do not specify how the monadic type constructors or binds are
defined. Also, the predicates for a bind specification are
very general and only some of those can be translated to Haskell. However, 
in an actual implementation, every monad constructor would give rise to a data type definition,
and every bind specification to an instance definition:
\begin{lstlisting}
[[m$/k$]] = data m a$_1$ $...$ a$_k$ t = ...  
[[b$:\forall\overline\alpha.\,\Phi \Rightarrow ($m1$,$m2$)\rhd\;$m3]] = instance [[$\Phi$]] => Bind m1 m2 m3 where (>>>=) = b    
\end{lstlisting}
A nice property of the Haskell translation together with Figure~\ref{fig:xlate-oml} is that
it defines a semantics for our system with a precise definition of
\emph{elaboration}, i.e. given our initial program,
how it is translated into a fully elaborated form where all
the binds are inserted. In particular, Haskell (or OML) will 
further elaborate where each \lstinline|Bind| constraint is
turned into an evidence term that is passed as an argument
at runtime.